\documentclass{article}

\usepackage{arxiv}

\usepackage[utf8]{inputenc} 
\usepackage[T1]{fontenc}    
\usepackage{hyperref}       
\usepackage{url}            
\usepackage{amsfonts}       
\usepackage{nicefrac}       
\usepackage{microtype}      
\usepackage{lipsum}		
\usepackage{graphicx}
\usepackage{natbib}
\usepackage{doi}

\usepackage{amsmath, amsfonts, mathtools, amssymb} 
\usepackage{amsthm} 
\usepackage[ruled,vlined]{algorithm2e} 
\usepackage{dsfont} 
\usepackage{xcolor} 
\usepackage{booktabs} 
\usepackage{float} 
\usepackage{enumitem} 
\usepackage{subcaption}
\usepackage{stfloats} 
\usepackage{xurl} 

\usepackage{balance}
\newtheorem{theorem}{Theorem}
\newtheorem{corollary}{Corollary}
\newtheorem{definition}{Definition}
\newtheorem{lemma}{Lemma}

\newcommand{\rel}{\mathsf{rel}}
\renewcommand{\P}{\mathsf{Prob}}
\newcommand{\D}{\mathds{P}}
\newcommand{\R}{\mathds{R}}

\newcommand{\util}{\mathcal{U}}
\newcommand{\unfair}{\mathcal{V}}
\newcommand{\tradeoff}{\mathcal{T}}
\newcommand{\group}{\mathsf{g}}
\newcommand{\score}{\mathsf{score}}
\newcommand{\bigO}{\mathcal{O}}

\newcommand{\docset}{\mathcal{D}}
\newcommand{\query}{Q}
\newcommand{\snippet}[1]{$\texttt{\footnotesize{#1}}$}

\DeclareMathOperator*{\argmax}{argmax} 
\newcommand{\Ex}{\mathbf{E}}
\DeclarePairedDelimiterX\Barg[1]{\lbrack}{\rbrack}{#1}
\DeclarePairedDelimiterX\multiset[2]{\lparen\lparen}{\rparen\rparen}{\genfrac{}{}{0pt}{}{#1}{#2}}

\title{Scoring is Not Enough: Addressing Gaps in\\Utility-Fairness Trade-offs for Ranking}

\author{ \href{https://orcid.org/0000-0003-1540-8511}{\includegraphics[scale=0.06]{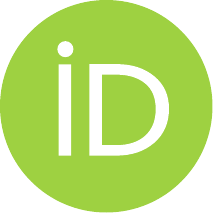}\hspace{1mm}Shubham Singh} \\
	Department of Computer Science\\
	University of Illinois Chicago\\
	Chicago, IL 60607 \\
	\texttt{ssing57@uic.edu} \\
	\And
	\href{https://orcid.org/0000-0002-7826-8555}{\includegraphics[scale=0.06]{orcid.pdf}\hspace{1mm}Ian A.~Kash} \\
	Department of Computer Science\\
	University of Illinois Chicago\\
	Chicago, IL 60607 \\
	\texttt{iankash@uic.edu} \\
    \And
	\href{https://orcid.org/0000-0002-6479-9769}{\includegraphics[scale=0.06]{orcid.pdf}\hspace{1mm}Mesrob I.~Ohannessian} \\
	Department of Electrical and Computer Engineering\\
	University of Illinois Chicago\\
	Chicago, IL 60607 \\
	\texttt{mesrob@uic.edu} \\
}

\renewcommand{\headeright}{}
\renewcommand{\undertitle}{}
\renewcommand{\shorttitle}{Scoring is Not Enough}

\hypersetup{
pdftitle={Scoring is Not Enough: Addressing Gaps in Utility-Fairness Trade-offs for Ranking},
pdfsubject={cs.LG, cs.CY},
pdfauthor={Shubham Singh, Ian A. Kash, Mesrob I. Ohannessian},
pdfkeywords={Ranking, Fairness, Tradeo-offs},
}

\begin{document}
\maketitle

\begin{abstract}
	Scoring functions are used to represent the relevance of individual documents. In modern information retrieval or recommendation systems, they are often learned from data and play a pivotal role in ranking sets of documents or items in a way that maximizes utility to a query or user. With the recent interest in algorithmic fairness, the success of scoring has naturally led to methods that learn scores that simultaneously trade off fairness and utility. In this work, we show that in stark contrast with utility-centric objectives, scoring is sub-optimal in achieving all utility-fairness trade-offs. We establish this with a series of counter-examples with a generic fairness formulation. We show that the issue persists whether we have a deterministic scoring function or a randomized one, or whether we measure fairness at the scope of a single query or across multiple queries. On the positive side, we empirically demonstrate that semi-greedy post-processing has the potential to achieve much better trade-offs, often approaching the ideal of exhaustive post-processing in a tractable way.
\end{abstract}



\section{Introduction}
\label{sec:intro}

Ranking is a basic task that has far-reaching applications in sociotechnical systems. Information retrieval and recommendation systems have provided the main impetus to understanding ranking, motivated by presenting documents in response to a query or recommending items to a user, in such a way as to maximize engagement. However, ranking is increasingly recognized to emerge in other diverse settings. For example, one may want to rank sites to determine the order in which they are inspected or to rank medical tests to decide the order in which they are administered. A primary utility often quantifies how well these ranking tasks are performed. With the diversity of applications, however, we increasingly recognize the importance of trading-off utility with secondary desirable objectives, such as fairness. To achieve such trade-offs in a data-driven fashion, the field has continued to adhere to a main paradigm: learning how to score documents and then, at inference-time, evaluating scores and ranking via sorting them. (See Section \ref{sec:related-work-mini} below for a quick overview, with a full survey of relevant work in Appendix \ref{sec:related-work}.) \looseness=-1

In this paper, we challenge this paradigm. More precisely, through a series of analytical, simulation, and real examples, \textbf{we show that scoring cannot cover the entire range of achievable trade-offs} and often falls considerably short of it. To make this thesis clear, we adhere to the language of information retrieval and to a generic family of binary document-group fairness metrics as a secondary objective, which can be instanced to some popular metrics via a choice of weights. The key insight is that, \textbf{unlike utility,  objectives such as fairness may not be decomposable into individual contributions} of documents.

We define scoring as assigning to each document, based on its features, a numerical score. This assignment may be deterministic or, in the case of some models such as Plackett-Luce,~\cite{plackettAnalysisPermutations1975,luceIndividualChoiceBehavior1959}, randomized. We assume to operate in the infinite-data learning regime, by substituting learning with two abstractions: (i) the true relevance and the group membership of documents are known and (ii) the ranking procedure can be determined based on the exact knowledge of how documents combine to form queries. In Section \ref{sec:setting}, we make these definitions and assumptions mathematically precise. Then, in Section \ref{sec:scorability}, to define the trade-off achieving aspirations we have of scoring, \textbf{we introduce the concept of scorability}. In Section \ref{sec:suboptimality}, we show that scoring is not enough to achieve optimal trade-offs, through analytical counterexamples to scorability and numerical examples. 

In Section \ref{sec:alt-scoring}, we give a positive spin to these findings. We look at an alternative to scoring: learning relevances without regard to the secondary objective and then, ex-post, accounting for it. Such post-processing, also known as re-ranking, needs to explore the space of rankings and can be intractable, except in very simple settings. Current techniques often resort to greed. In Section \ref{sec:evaluation}, relaxing greed, \textbf{we empirically demonstrate that beam-search can efficiently cover achievable trade-offs} much more than scoring and also better than greed alone.

Fair learning-to-rank via scoring is an in-processing approach and, as such, is touted to offer more flexibility than post-processing~\citep{singhPolicyLearningFairness}. Thus, the main message of this paper is to alert the community to the fact that, when striving to trade off a decomposable utility and a non-decomposable objective such as fairness, \textbf{the nature of scoring could make it limiting}. 
Figure \ref{fig:overall-fig} gives an overview of the key ideas and results of the paper.

\begin{figure}
    \centering
    \includegraphics[width=\linewidth]{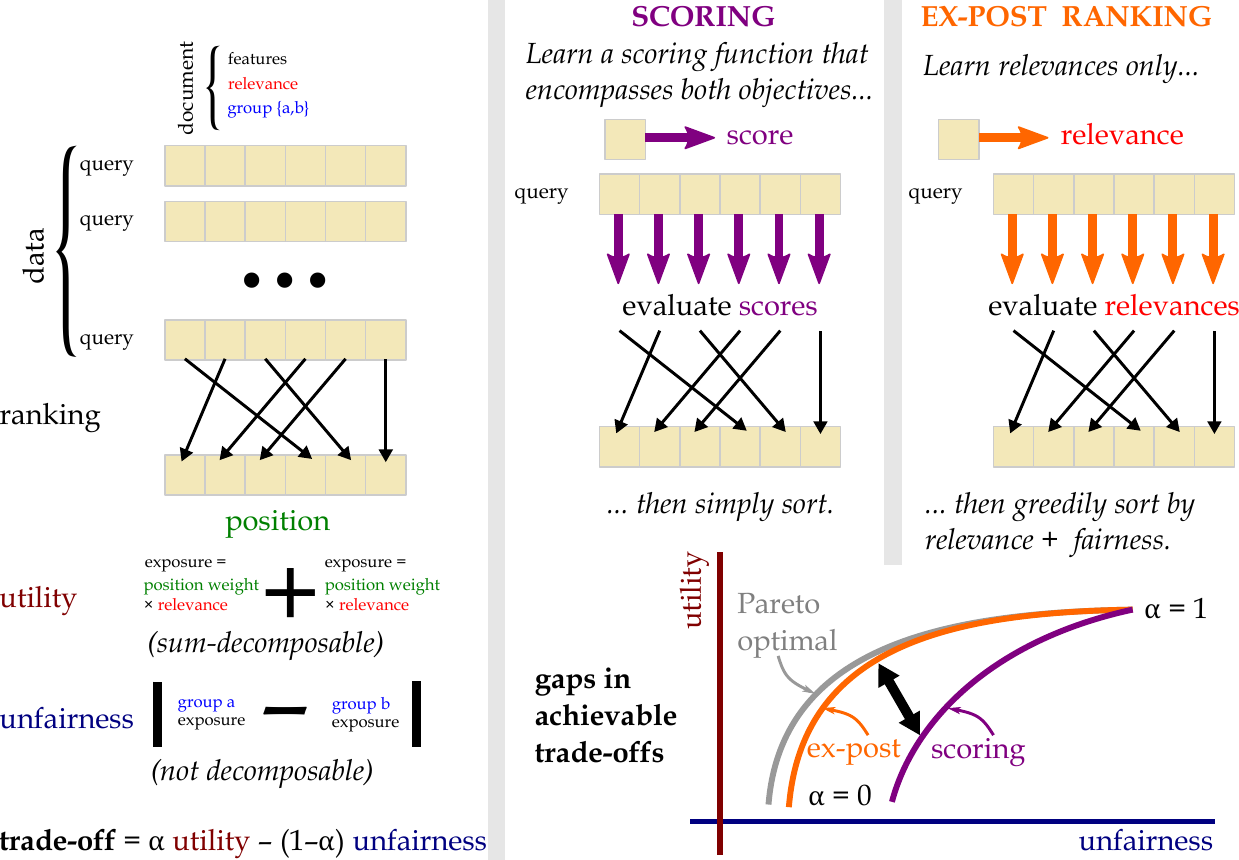}
    \caption{Overview of the paper, highlighting key ideas about the limitations of scoring for utility-fairness trade-offs}\vspace{-16pt}
    \label{fig:overall-fig}
\end{figure}



\vspace{-6pt}
\subsection{Related Work}
\label{sec:related-work-mini}
\vspace{-6pt}
To optimize the utility-fairness trade-off, we utilize the low-variance Plackett-Luce gradient estimator by \citeauthor{gadetskyLowvarianceBlackboxGradient2019}~\citep{gadetskyLowvarianceBlackboxGradient2019}. This aligns our work with recent reinforcement learning approaches for scoring~\citep{singhPolicyLearningFairness,ge2022toward}.
We focus specifically on document-group fairness via interaction disparity per query~\citep{gorantlaOptimizingGroupFairPlackettLuce2023}. This distinguishes our scope from works targeting user fairness~\citep{greenwood2024user} or accumulated cross-query interactions~\citep{memarrastFairnessRobustLearning2021,ovaisiFairnessInteractionRanking2024}.
While some post-processing methods enforce fixed fairness~\citep{geyik2019fairness,gorantlaOptimizingGroupFairPlackettLuce2023}, others address utility-fairness trade-offs using greedy heuristics~\citep{gao2020toward,ovaisiFairnessInteractionRanking2024}. Crucially, while \citeauthor{ovaisiFairnessInteractionRanking2024} fits our fairness framework, prior literature has not recognized that scoring itself is inherently suboptimal. \looseness=-1

Our work distinguishes itself by explicitly highlighting the inadequacy of scoring-based ranking methods for achieving optimal utility-fairness trade-offs. A detailed discussion of the literature is given in Appendix~\ref{sec:related-work}.
\section{Problem Setting} \label{sec:setting}



We start by defining the general ranking setting in which we operate. We use the following notation. Let the set of possible documents be denoted by $\docset$ and individual documents be denoted by $d \in \docset$. We associate with each document a relevance $\rel(d) \in [0,1]$. We assume that each query consists of a set
$\query$ of $m$ documents from $\docset$.

Relevances, in general, depend on the query. By forgoing this dependence, we are restricting to a specific case (no personalization). However, when a limitation is present in a specific case it is also present in the general case in that instance. Thus the conclusions of the paper hold generally.

\paragraph{Deterministic and Random Ranking} A \emph{deterministic ranking} $\sigma_\query$ is defined as a permutation of the $m$ documents in $\query$, in a way that may depend on $\query$. $\sigma_\query(i) = d$ is interpreted as placing document $d\in \query$ at position $i$. We assume that no selection is made, that is all documents in the query are ranked. More precisely, for each $\query$, $\sigma_\query$ is equivalent to a bijective map $[m]\rightarrow \query$.
We also allow for \emph{randomized ranking}. In this case, $\sigma_\query$ is a random map defined through a distribution, conditional on $\query$, over all $m!$ permutations. In other words, given $\query$, $\sigma_\query $ is sampled from this conditional distribution. Although $\sigma_\query$ always depends on $\query$, whenever $\query$ is explicitly specified we drop the subscript $_\query$ and simply write $\sigma$.

\paragraph{Utility} Utility, in general, is a function from $(\sigma, \query)$ to $\R_+$. In most ranking applications, however, this utility decomposes as a sum. A document at position $i$ delivers a utility based on its relevance as well as a factor that depends only on position, captured by non-increasing weights $w_i$. The \emph{utility} of a deterministic ranking can then be given as:
\begin{equation} \label{eq:utility}
    \util(\sigma, \query) = \sum_{i} w_i \rel(\sigma_\query(i)).
\end{equation}
\vspace{-6pt}
Examples:
\begin{itemize}
\item[] \textbf{Recommendation systems} If we use $w_i = 1/(\log_2 i + 1)$ to capture position bias, utility becomes the discounted cumulative gain (DCG). Here, $w_i \rel$ represents an interaction probability, thus maximizing utility boosts interaction.
\item [] \textbf{Inspection sites} Documents could refer to inspections sites, with $\rel(d)$ the probability of a failure occurring at site $d$, and $\sigma_\query$ the order in which sites are visited to correct failures. Then, choosing $w_i = - (i-1)/m$ and maximizing utility minimizes average uncorrected failure time.
\end{itemize}
\vspace{-6pt}
For randomized ranking, because of its linearity, the notion of utility extends in a straightforward way to expected utility. We keep the notation the same, and distinguish the two instances based on context:
\begin{equation} \label{eq:expected-utility}
    \util(\sigma, \query) = \sum_{i} w_i \Ex\left[\rel(\sigma_\query(i)) \vert \query \right].
\end{equation}
Note that the expectation is over random rankings, while the query $\query$ is fixed.
\vspace{-6pt}
\paragraph{Unfairness} In addition to utility, we are concerned with giving documents a fair interaction through ranking. For the purpose of defining fairness, we think of $\docset$ as being partitioned into two groups $G_a$ and $G_b$, and use $\group(d) \in \{a, b\}$ as the group-membership function. We capture fairness, or rather \emph{unfairness}, by measuring the disparity of relative interaction between the two groups:
\begin{equation} \label{eq:unfairness}
    \begin{multlined} 
    \unfair(\sigma, \query) = \left| \frac{1}{|G_a|} \sum_{i:\group(\sigma_\query(i)) = a} \!\!\!w_i \rel(\sigma_\query(i)) \right.  \left. - \frac{1}{|G_b|} \sum_{i:\group(\sigma_\query(i))=b} \!\!\!w_i \rel(\sigma_\query(i)) \right|  
    \end{multlined}
\end{equation}
By convention, if either group is not present in $\query$, we assume fairness is vacuously achieved and set $\unfair=0$.
Examples:
\begin{itemize}
\item[] \textbf{Recommendation systems} If the documents are interview candidates, utility captures consideration probability, and the groups refer to two racial groups, this disparity would be the difference of per-candidate consideration across these groups.
\item [] \textbf{Inspection sites} If $G_a$ and $G_b$ are two geographical inspection zones, this would be the disparity in failure time correction across zones. 
\end{itemize}

Other variants of \eqref{eq:unfairness} can be justified in various contexts, e.g., normalizing by the total relevance of each group instead of their sizes, using only $w$ in the sum instead of $\rel$, etc. We choose \eqref{eq:unfairness} because it is a common disparity measure and gives a concrete instance to illustrate the phenomenon that we aim to shed light on. This phenomenon, however, is likely not limited to this choice, as we discuss later. 

For randomized ranking, some authors (e.g., \cite{singhPolicyLearningFairness}) aggregate the disparity over the random rankings within the absolute value in \eqref{eq:unfairness}. This can be thought of as corresponding to a long session in which multiple rankings may be produced for a single query. However, since \eqref{eq:unfairness} measures per-query unfairness and since during a single query only a single ranking is sampled, we adhere to the more natural choice of aggregating outside of the absolute value, i.e., a short session:
\begin{equation} \label{eq:expected-unfairness}
  \begin{multlined} 
  \unfair(\sigma, \query) = \Ex\left[ \left| \frac{1}{|G_a|} \sum_{i:\group(\sigma_\query(i)) = a} \!\!\!w_i \rel(\sigma_\query(i)) \right. \right. \left. \left. - \frac{1}{|G_b|} \sum_{i:\group(\sigma_\query(i))=b} \!\!\!w_i \rel(\sigma_\query(i)) \right| ~\Bigg\vert \query \right],
   \end{multlined}
\end{equation}
Again, the expectation is over random rankings while the query $\query$ is fixed, and we keep the notation the same and disambiguate based on context. Readers interested in the effect of session length, albeit in a slightly different setting, can be referred to~\citep{vardasbi2022probabilistic}.




\paragraph{Problem Instances} We call the triplet $(\docset, \rel, \group)$ a \emph{problem instance}. The space of utility and unfairness pairs $(\util,\unfair)$ is spanned through the choice of the weights $w$. Eqs. \eqref{eq:utility} and \eqref{eq:unfairness}, or Eqs. \eqref{eq:expected-utility} and \eqref{eq:expected-unfairness}, give the functional form of $(\util,\unfair)$ in the deterministic, or randomized, ranking cases respectively. These functional forms then become concrete for each problem instance $(\docset, \rel, \group)$.

\paragraph{Trade-offs} With the primary objective of utility and the secondary objective of fairness, we are interested in striking a good trade-offs between utility \eqref{eq:utility} and unfairness \eqref{eq:unfairness}. Trade-offs at level $\alpha$ can be defined through the combined objective function:
\begin{equation} \label{eq:trade-off}
    \tradeoff_\alpha(\sigma, \query) = \alpha \ \util(\sigma, \query) - (1-\alpha) \ \unfair(\sigma, \query),
\end{equation} 
where the minus in the second term refers to the fact that we strive to increase utility, but to decrease unfairness.

\paragraph{Deterministic Scoring} Scoring, in a typical learning-to-rank framework, produces a function $\score(d)~:~\docset \to \R$ using either pairwise or list-wise data, as overviewed in Appendix \ref{sec:related-work}. At query-time, $\score(d)$ is computed for each document in the query $\query$ and the ranking is produced by sorting the documents in decreasing score order. The score-induced ranking can then be defined as:
\begin{equation} \label{eq:score-ranking}
    \sigma_{\query,\score} = \mathsf{sort}_\downarrow(\query,\score).
\end{equation}
With limited data, $\score$ typically does not have direct access to $\rel$, but rather depends on features of the document, potentially including its group membership. To narrow down on the issue at hand, consider instead the \textbf{infinite-data regime} where $\score$ knows exactly $\rel(d)$ as well as $\group(d)$. Since this information is sufficient to evaluate both \eqref{eq:utility} and \eqref{eq:unfairness}, and thus \eqref{eq:trade-off}, there is no loss of capability in restricting $\score$ to be an arbitrary function of these two attributes. Thus, moving forward, $\score(d)$ is always of the form
\[
    \score\left(\rel(d),\group(d)\right)~:~[0,1] \times \{a,b\} \to \R.
\]
Note that many documents may have the same (relevance, group) pair, and we can think of this pair as a document \emph{type}. In what follows, even if documents themselves do not repeat in a query, document types may. Lastly, to avoid ties, we assume $\score$ never assigns distinct document types in $\docset$ the same score.\footnote{This is not a restrictive assumption, because ultimately we need to break ties, and this can only be done via either deterministic or stochastic perturbation. If the perturbation is deterministic, then it is equivalent to the assumption. If the perturbation is stochastic, then we simply jump to the case of randomize scoring, which doesn't make this assumption.}


\paragraph{Randomized Scoring} Deterministic scoring is widely used, but a softer variant gives more flexibility. In this case, $\score(d)$ is a random variable drawn independently from other documents in $\query$, from a distribution $f_d$ on $\R$ that depends uniquely on $d$. In the infinite-data regime, $f_d$ depends uniquely on $\rel(d)$ and $\group(d)$. In this case, the score-induced ranking $\sigma_{\query,\score}$ of \eqref{eq:score-ranking} is a randomized ranking. The main practical case of randomized scoring that we consider is the Plackett-Luce (PL) model, a ubiquitously used model for list-wise ranking~\citep{xiaPlackettLuceModelLearningtorank2019, maLearningtoRankPartitionedPreference2021,gadetskyLowvarianceBlackboxGradient2019,gorantlaOptimizingGroupFairPlackettLuce2023,oosterhuisComputationallyEfficientOptimization2021, singhPolicyLearningFairness}. The best way to relate PL to randomized scoring is to note that it is equivalent to learning a deterministic score $h(d)$ first, and then to randomizing it by adding independent standard Gumbel noise to each:
\[
    \score_\mathsf{PL}(d) = h(d) + \mathsf{Gumbel}.
\]
This then induces a randomized ranking:
\[
    \sigma_{\query,\mathsf{PL}} = \mathsf{sort}_\downarrow\left(\query,\score_\mathsf{PL}\right).
\]

As proposed by \citet{plackettAnalysisPermutations1975} and \citet{luceIndividualChoiceBehavior1959}, an alternative way of describing PL is as a sequential generation of the ranking itself, through the probability of a document being ranked at the next position $i$, conditioned on the remaining documents up till that point. 
\begin{equation} \label{eq:plackett-luce}   
    \P\left(\sigma_{\query,\mathsf{PL}} = \sigma \vert \query \right) = \prod_{i=1}^{m} \underbrace{\left(\frac{\exp({h(\sigma(i))}}{\sum_{j=i}^{m} \exp({h(\sigma(j))}}\right)}_{\P\left(\sigma(i)|\query\setminus (\sigma(k))_{k<i}\}\right)}.
\end{equation}
In learning to rank, typically, $h$ is parametrized using, for example, a linear regressor or neural network that depends on the features of $d$. The resulting ranking sampling is often referred to as a stochastic ``ranking policy'' $\pi_{\query}$. In the infinite-data regime, we can take $h$ to be simply a function of the relevance and group membership, $h(d)=h(\rel(d),\group(d))$.

\section{Scorability} \label{sec:scorability}

The central question of this work is: to strike good trade-offs between utility \eqref{eq:utility},\eqref{eq:expected-utility} and unfairness \eqref{eq:unfairness},\eqref{eq:expected-unfairness}, i.e., achieve a low value of the trade-off at a desired level $\alpha$ \eqref{eq:trade-off}, is it sufficient to rely on scoring functions?

Our expectations from a scoring function, however, are often implicit. We posit that the popularity of scoring functions suggests the following plausible expectation:

\fbox{\parbox{0.97\linewidth}{Given enough data, we should be able to learn a scoring function that achieves the best possible utility-fairness trade-off.}}

To capture the spirit of this informal statement without being distracted by the details of the learning procedure, the following definition gives the scoring function infinite data by allowing it to depend on the query-generating distribution.

\begin{definition}[Scorablility] \label{def:weakly-scorable}
    A pair of utility and unfairness functions $(\util,\unfair)$ is called deterministic (or, respectively, randomized) \emph{$\alpha$-scorable} if $\forall$ problem instances $(\docset, \rel, \group)$ and distributions $\D$ on $\docset$, $\exists$ a deterministic (or, respectively, randomized) $\score$ such that if the documents in $\query$ are sampled $\sim_{\mathrm{i.i.d.}} \D$, then ranking them by score using $\sigma_{\score} = \mathsf{sort}_\downarrow(\query, \score)$ achieves
    \begin{equation} \label{eq:weak-optimality}
        \Ex_\query\left[\tradeoff_\alpha\left(\sigma_{\score}, \query\right) \right] =  \max_{\sigma_\query} ~\Ex_\query\left[\tradeoff_\alpha\left(\sigma_Q, \query\right) \right]
    \end{equation}
    where the maximization is over deterministic (or, respectively, randomized) rankings. Equivalently, scorability means that:
    \[
        \sigma_\score \in \argmax_{\sigma_\query} ~\Ex_\query\left[\tradeoff_\alpha\left(\sigma_Q, \query\right) \right].
    \]
\end{definition}
\vspace{-0.5em}
It is important to note that the maximization in Eq. \eqref{eq:weak-optimality} is over $\sigma_\query$, the space of rankings or equivalently (deterministic or random) mappings from $\query$ to permutations, and not over a single (deterministic or random) permutation.

By allowing $\score$ to depend on the problem instance and the distribution of queries, Definition \ref{def:weakly-scorable} abstracts the notion that the scoring function may be learned given sufficient data from this instance. The order of quantifiers ``$\forall$ $(\docset, \rel, \group)$ and $\D$, $\exists ~\score$ s.t. $\forall \query \sim \D$'' is extremely important for this abstraction to be meaningful. It can be understood as ``given enough data, we can learn/build a scoring function, such that when new queries arrive'' we can rank them well by scoring then sorting. This adheres to the typical learning-to-rank framework, where the scoring function acts on individual documents and is fixed after the learning process.


\section{Suboptimality of Scoring} \label{sec:suboptimality}
\subsection{Deterministic Case} \label{sec:deterministic-scorability}

We first consider the deterministic case and show that there are instances where scorability is not possible. In other words, in these instances, deterministic scoring does not achieve the best possible tradeoff.

\begin{theorem}[Counterexample to Deterministic Scorability] \label{thm:weak_scorable}
    Let $\mathcal U$ and $\mathcal V$ be given by Eqs. \eqref{eq:utility} and \eqref{eq:unfairness} with some monotonically non-increasing $w$ such that $w_i+w_{m-i+1}$ is not constant for $i=1,\cdots,m-1$. Then, for any $\alpha<1$, there exists a problem instance $(\docset, \rel, \group)$ and a distribution $\D$ on $\docset$, such that for all deterministic scoring functions $\score$, that can depend on the instance and the distribution, Eq. \eqref{eq:weak-optimality} fails. Therefore, this choice of $(\mathcal U, \mathcal V)$ is not scorable.
\end{theorem}

The detailed proof of Theorem \ref{thm:weak_scorable} can be found in Appendix~\ref{app:proofs}, but we offer here the intution behind it.
To show that the utility and fairness pair is not scorable, we construct a problem instance with only one document type (same relevance) from each group. This means that only two scoring functions are possible: rank group $G_a$ first or group $G_b$ first. We show that no matter which is chosen, we can design an adversarial query $Q$ to contradict the optimality of that scoring function. That is, no matter how a scoring function is chosen, there exists some query $Q$ where it will fail to reach the Pareto frontier. \textbf{The reason we can design such a $Q$ is the fact that, whereas the score depends on individual documents only, fairness depends on the makeup of the whole query.} Designing $Q$ like this is not enough, however, because the definition of scorability allows the choice of score to depend on the distribution over $Q$. So, if we deterministically chose one of our adversarial queries, one could simply choose the other scoring function. The solution is to choose a distribution over queries that guarantees {\em both} adversarial queries are chosen with positive probability.

In Appendix~\ref{app:proofs}, we give a corollary to Theorem \ref{thm:weak_scorable} which explicitly quantifies the optimality gap under an additional (convexity) condition on the weights $w$. We also illustrate an example instance of this scenario, with the position bias weights $w_i=1/(\log_2 i + 1)$, where we evaluate the theoretical bound on the gap as well as calculate the actual gap via Monte Carlo simulations. The gap persists when the two types of documents do not have the same relevance, and its general behavior (e.g., being much smaller for large $\alpha$) parallels all of our other experiments, including with real data. \looseness=-1
\vspace{-6pt}
\subsection{Randomized Case} \label{sec:randomized-scorability}
\vspace{-6pt}
Having identified that deterministic scorability fails for the family of utility and unfairness objectives considered, we now demonstrate that randomization does not necessarily help circumvent the trade-off gap. The proof relies on the fact that it is always possible to improve on randomized scoring using a randomized ranking that mimics it always, except in the cases highlighted by the counterexample construction in Theorem \ref{thm:weak_scorable}.
\begin{theorem}[Counterexamples to Randomized Strong/Weak Scorability] \label{thm:rand_scorable}
    Let $\mathcal U$ and $\mathcal V$ be given by Eqs. \eqref{eq:expected-utility} and \eqref{eq:expected-unfairness} with some monotonically non-increasing $w$ such that $w_i+w_{m-i+1}$ is not constant for $i=1,\cdots,m-1$. Then, for any $\alpha<1$,
        there exists a problem instance $(\docset, \rel, \group)$ and a distribution $\D$ on $\docset$, such that for all randomized scoring functions $\score$, that can depend on the instance and the distribution, Eq. \eqref{eq:weak-optimality} fails. Therefore, this choice of $(\mathcal U, \mathcal V)$ is not scorable.        
\end{theorem}
\vspace{-0.5em}
In order to show that this is not just a concern with theoretical constructions, we next consider the widely used Plackett-Luce model and  construct a tractable counterexample instance where the Pareto frontier can be calculated, then numerically optimize the Plackett-Luce model (hence the randomized scoring) and demonstrate its suboptimal trade-offs. \looseness=-1

We let weights to be the standard position bias $w_i = 1 /(\log_2(i)+1)$. We choose $\D$ to have $k = 5$ documents, with the following composition $\left(\rel(d),\group(d)\right) = (1,a), (2,b), (3,a), \allowbreak (4,b), (5,a)$. We sample, i.i.d., queries $\query$ consisting of $m = 8$ documents. This is the same as the synthetic data set described in Section \ref{sec:evaluation}, which also contains details about how the PL model is solved. Primarily, to calculate gradients directly is intractable and one often resorts to the log-trick, as in the REINFORCE algorithm, However, a variance reduction technique proposed by \citet{grathwohlBackpropagationVoidOptimizing2018} has been demonstrated to work well for training PL models~\citep{gadetskyLowvarianceBlackboxGradient2019}, and is particularly reliable for this simple counterexample. 

\begin{figure*}[h] 
\centering
\begin{subfigure}[t]{0.4\textwidth}
    \includegraphics[width=\linewidth]{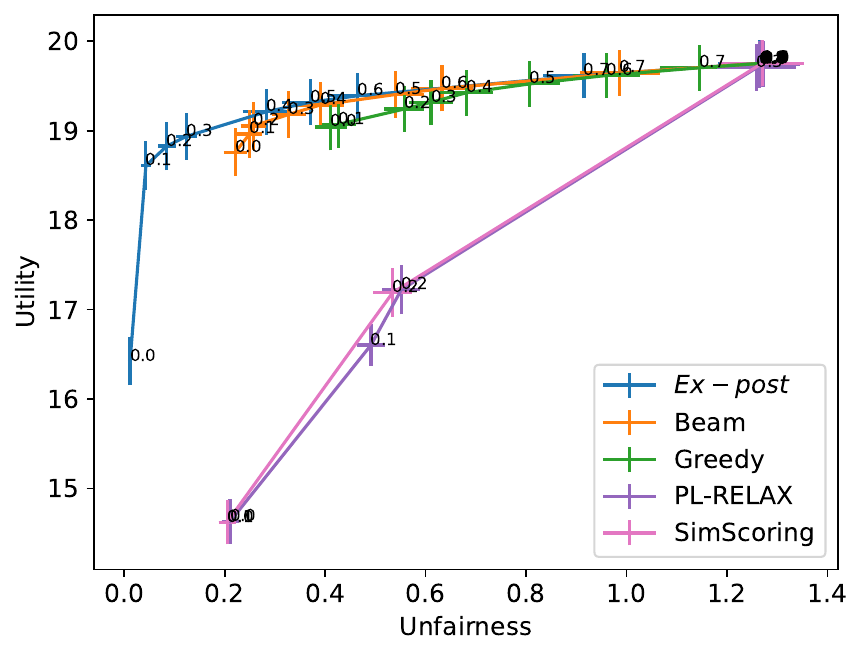}
    \caption{Synthetic Dataset}
    \label{fig:alt-scoring-synthetic}
\end{subfigure}
\hspace{3em} 
\begin{subfigure}[t]{0.4\textwidth}
    \includegraphics[width=\linewidth]{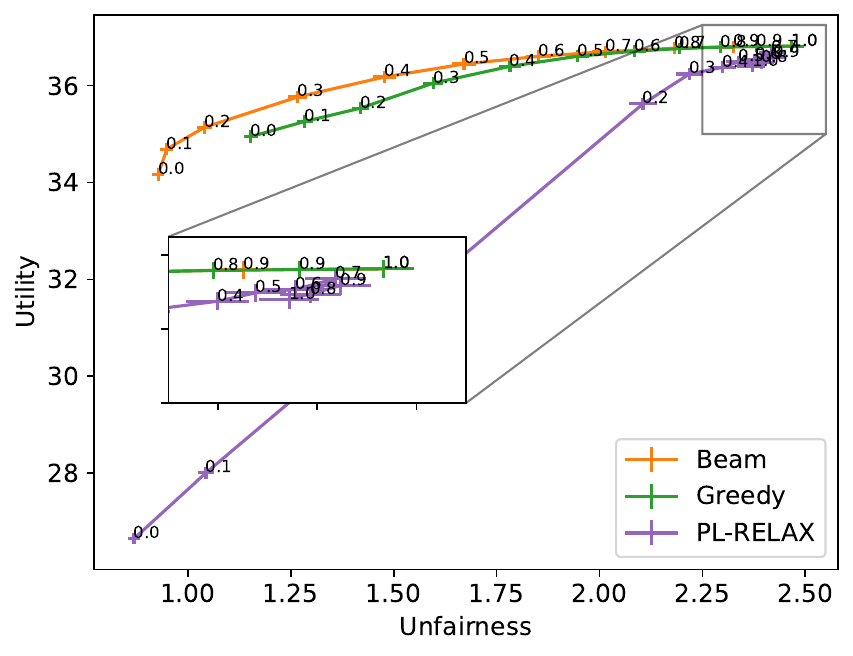}
    \caption{COMPAS Dataset}
    \label{fig:alt-scoring-compas}
\end{subfigure}
\caption{Utility and unfairness values for $\alpha = \{0, \ldots, 1\}$ achieved by beam search (orange), greedy (green) and the randomized scoring-based trained PL model (purple) on these datasets. Unlike the real-world COMPAS dataset, the synthetic data is tractable, and we also include the brute-force optimal \textit{ex-post} (blue) and best deterministic scorer (pink).}
\vspace{-12pt}
\end{figure*}


Further details of the experiments are included as part of Section \ref{sec:evaluation}, but we illustrate the resulting suboptimality in Figure \ref{fig:alt-scoring-synthetic} below. The top (navy) curve is the exact (exhaustive) ex-post optimal solution. We can see clearly that the PL model bottom (purple) curve is unable to achieve good trade-offs in comparison. We also include an exhaustive search of all (deterministic) scores (pink), which shows that the randomization of Plackett-Luce does not give it an advantage in this case, as their tradeoff curves effectively overlaps (within simulation precision). In Section \ref{sec:alt-scoring}, we propose and explain the alternative approaches (green and orange) in the same figure.

\vspace{-6pt}
\subsection{Across-Query Fairness} \label{sec:across-query}
\vspace{-6pt}
The notion of unfairness given by Eq. \eqref{eq:unfairness} is per-query. Namely, we measure the discrepancy between groups at the level of individual queries. While this is reasonable in many applications, e.g., suggesting male and female applicants in response to a single job posting, there are other instances where fairness is better measured across queries, e.g., whether certain movie categories are reaching a wide enough audience. In Appendix \ref{app:across-query}, we demonstrate that even in that setting, scoring can be suboptimal.

\section{Alternatives to Scoring} \label{sec:alt-scoring}

Having seen the theoretical shortcomings of scoring in various contexts, be it deterministic, randomized, per query $D$ or across queries drawn from a distribution $\D$, we now turn our attention to potential alternatives to simultaneously highlight that these shortcomings are present even in practice and, possibly, to overcome them. In particular, we explore and propose ex-post methods that can achieve better trade-off values than scoring functions. 

In Section \ref{sec:evaluation}, when tractable, we compare these approximate ex-post approaches to exact brute-force post-processing, which evaluates all the permutations for a query $D$ and finds the rankings that achieve the best trade-offs. This motivates us to think of the algorithms presented here as proxies to these Pareto optimal rankings when they are intractable, thus providing a clear backdrop to highlight the limitations of scoring. Thus, while we invite more attention to the richness of post-processing and are excited about some of the novelties in our methods, such as the use of beam search in fair ranking, our emphasis remains on understanding scoring.

In what follows, we assume that we know the true relevance of the documents or that we can learn them. 
\vspace{-12pt}
\subsection{Greedy Approach} \label{sec:alt-greedy}
\vspace{-6pt}
We first start with an ex-post approach that is a representative of the methods covered in Section \ref{sec:related-work}, particularly \citep{ovaisiFairnessInteractionRanking2024}, which is itself inspired by~\citep{geyik2019fairness}. The approach also has elements in common with~\citep{gao2020toward}. We emphasize that this is only meant to \emph{represent} these methods, not to \emph{compete} with them, similarly to how the Plackett-Luce model is meant to represent, and not to compete with, scoring approaches such as Fair-PG Rank~\citep{singhPolicyLearningFairness}.

In the first step in this greedy approach, we split the documents in $D$ by groups, $\{a,b\}$, and within each group, sort the documents in decreasing order of their relevance.
The group-wise sorted document sets let the greedy approach efficiently decide on the most suitable document for a better utility-fairness trade-off by considering only the top remaining document from each group.

Next, we initialize an empty list of ranked documents. For each position, $i$, we try adding the top document from each group, $\{a,b\}$, and calculate the trade-off with the desired $\alpha$. The document from the group that achieves the maximum trade-off value is placed at position $i$ and removed from the group-wise sorted document sets. The trade-off value at any position $i$ is computed for the previously selected documents and the one considered at position $i$, given as:
\begin{align}
    \sigma_{ D}(i) = \arg\max_{d \in \{d^{*}_a,d^{*}_b\}} \tradeoff\left(
            \sigma(\Barg{0, \dots, i-1}) \cup \{d\}
        \right) 
\end{align}
where $d^{*}_\group$ is the top (highest relevance) document from group $\group$, and $\tradeoff$ represents the (truncated) trade-off function. We repeat this process for all positions, and the final result is a ranked list of all documents.

Without the unfairness term in the trade-off, this greedy approach would exactly optimize utility. Indeed, for higher $\alpha$ values, where utility is given more weight in the trade-off, we observe that this performs well. However, since unfairness of a ranking depends on the relative positions of other documents, the one-step decision in the greedy approach is too myopic and can get stuck in local minima. This is why performance is expected to deteriorate with lower $\alpha$ values.
\vspace{-6pt}
\subsection{Beam Search} \label{sec:beam}
\vspace{-6pt}
To improve on greed, we propose a novel contribution: to adapt ideas from beam search~\citep{lowerreHARPYSPEECHRECOGNITION1976}, commonly used in NLP procedures~\cite{freitagBeamSearchStrategies2017} into the previous method. Beam search is similar to the greedy approach, but relaxes the greed by considering more than one possible realization. The algorithm can keep multiple lists of exploration (a beam). The number of lists is denoted by $B$, called the beam width, and is an input parameter of the algorithm. While fairly standard, the complexity here comes from the fact that each list in the beam is like a parallel universe, where we need to book-keep everything (which documents have been ranked and the remaining documents in each group, sorted in order of decreasing relevance.) At each stage, we expand the beam by considering the next candidate from each group, then we prune the beam by evaluating all the tradeoffs achieved and keeping only the top-$B$. We repeat this process until all the remaining documents in the set have been ranked. In the end, the ranking with the best trade-off among the last $B$ rankings is returned. We describe a concrete implementation of this beam search for ranking trade-offs in Algorithm~\ref{alg:beam-search}. A detailed mathematical description and a schematic illustration of this algorithm can be found in Appendix \ref{app:beam-details}. Both beam search and greedy are extremely efficient. They both need $\bigO(m \log m)$ time to initially sort. Greedy then needs time $\bigO(m)$ to traverse, while beam search in addition needs $\bigO(m \ B \log B)$ to traverse while expanding and pruning the beams.
\vspace{-6pt}
\section{Evaluation} \label{sec:evaluation}
\vspace{-6pt}
We evaluate both scoring via listwise Plackett-Luce ranking and our ex-post approaches (greedy and beam search) on one synthetic dataset and two real-world datasets. We show that, across all ranges of $\alpha$, our approaches achieve considerably better trade-offs than scoring. In the case of synthetic data, the results are also compared to the optimal trade-off, which ex-post approaches but scoring does not.
\vspace{-6pt}
\subsection{Data} \label{sec:eval_data}
\vspace{-6pt}
Our evaluation is done on four datasets---a synthetic dataset, COMPAS dataset~\cite{mattuMachineBias}, and the German Credit dataset~\cite{statlog_(german_credit_data)_144}, and the MovieLens 100K dataset (see Appendix~\ref{sec:other-datasets} for the last two). To construct the synthetic dataset, we select $k=5$ unique documents with relevances $\{1,...,5\}$ and alternating group membership $\{a, b, a, b, \cdots\}$. For each query $D$ on the $k$ documents, we sample $m=8$ random rankings. We sample 50,000 such queries to constitute a synthetic training set. The small values of $k$ and $m$ are intended to make the exact exhaustive ex-post optimal solution tractable in this case.


\setlength{\textfloatsep}{0pt}
\setlength{\floatsep}{0pt}
\begin{algorithm}[ht!]
\small
\SetAlgoSkip{0pt}
\caption{Beam Trade-off Search }
\label{alg:beam-search}
\KwIn{
Problem instance ($D$, $\rel$, $\group$),
Beam width $B$, Size of query $m=|D|$}
\KwOut{Trade-offs for varying $\alpha$}
\ForEach{$\alpha \in [0, 1]$ with step size $0.1$}{
    $\tradeoff_\alpha = \alpha \cdot \util(\sigma, D) - (1-\alpha) \cdot \unfair(\sigma, D)$\;
    
    Sort $D$ based on $\rel$ in descending order\;
    Group sorted $D$ into dictionary \snippet{desc\_group\_dict} based on $\group$\;
    Initialize the beam with top-ranked documents from each group\;
    Set \snippet{beam\_candidates} to the $B$ best initial choices based on their trade-offs\;
    \For{each position $i = 1, \dots, m$}{
        Initialize empty sets for \snippet{trade\_off\_list} and \snippet{seek\_indices}\;
        \For{each beam candidate $w \in $ \snippet{beam\_candidates}}{
            Collect remaining unprocessed documents for each group from \snippet{desc\_group\_idx\_dict}\;
            Get the top documents from each group sorted by relevance\;
            Combine most relevant documents into a single list\;
            \For{each document $d \in $ \snippet{remaining}}{
                Compute trade-off for appending $d$ to beam candidate $w$\;
                Append the trade-off and the updated beam to \snippet{trade\_off\_list} and \snippet{seek\_indices}\;
            }
        }
        Select the top $B$ beams with the lowest trade-off values\;
        Update \snippet{beam\_candidates} accordingly\;
    }
    Select the final beam in \snippet{beam\_candidates} as the candidate with the lowest trade-off\;
    Compute final utility and unfairness for the selected beam\;
}
\end{algorithm}

We  adapt the COMPAS dataset for a ranking task (as done in~\cite{singhPolicyLearningFairness}) by following a similar process to the synthetic dataset. This dataset contains 4,534 rows with 13 attributes that are used to predict if a defendant is likely to recidivate (commit a crime when release on parole). The attribute \texttt{violent\_decile\_score} assigns a score between 1 and 10 to each defendant, which we treat as the true relevance in this case. A query is created by sampling 10 documents from the dataset, where the ratio of recidivism likelihood (\texttt{two\_year\_recidivous}) is 4:1. We perform an 80-20 train-test split on the original documents and sample 10,000 queries from the train split and 2,000 from the test split. We treat the \texttt{race} attribute as protected, and use the two groups, \texttt{African-American} and \texttt{Caucasian}, for unfairness.  

\vspace{-6pt}
\subsection{Experiments}
\label{sec:experiments}
\vspace{-6pt}
We design the experiments to measure how trade-offs between utility and unfairness are navigated by different methods at different $\alpha$ values.

\textbf{\textit{Ex-Post.}} The \textit{ex-post} ranking iterates over all possible permutations given the documents in a query and finds the optimal given an $\alpha$ value. Naturally, this is an expensive operation and takes $\bigO(n!)$ time. Therefore, we perform this evaluation only for the synthetic dataset with a smaller query size of $m=8$. Using this baseline, we get the best possible ranking for each query, which allows us to contextualize the performance of the other approaches.

\textbf{Greedy and Beam Search.} These proposed alternative approaches are elaborated in Section \ref{sec:alt-scoring} and in Algorithm \ref{alg:beam-search}. They are approximate ex-post solutions that require knowledge of relevances, either exact or learned.

\textbf{PL Model.} We use PL-RELAX, the low-variance gradient estimator for the Plackett-Luce distribution, proposed by \citet{gadetskyLowvarianceBlackboxGradient2019}, to train on our datasets. The proposed estimator is unbiased and offers low variance compared to the other gradient estimators, such as REINFORCE. The PL model is a listwise ranking algorithm that learns a real-number preference over the documents and samples a ranking that optimizes a loss function. In our case, the loss function is the trade-off between utility and unfairness, Eq.~\eqref{eq:trade-off}.

\textbf{SimScoring.} For the synthetic dataset, we also simulate the best possible scoring methods by evaluating all possible (deterministic) scorings. This is equivalent to looking at all possible permutations of the document types, which is fairly tractable at this small scale. \looseness-1

\vspace{-6pt}
\subsection{Results}
\vspace{-6pt}
\textbf{Synthetic Dataset.} We plot the results for the \textit{ex-post}, beam search, greedy, SimScoring, and the PL-model on the synthetic dataset in Fig.~\ref{fig:alt-scoring-synthetic}. Each line corresponds to the trade-offs offered by one of the approaches for $\alpha \in \{0, \ldots, 1\}$. 
The overlap between SimScoring and the PL-RELAX model provides evidence of optimal convergence of the PL-model.
The greedy approach performs better than the PL model with gradient estimates (shown in purple, labeled \texttt{PL-RELAX}) and for $\alpha > 0.5$ gets close to the Pareto frontier achieved by the \textit{ex-post} approach. One reason that explains the better performance of the greedy approach at higher $\alpha$ values is as follows. When we prefer utility over fairness, since our choice of documents across groups is sorted in a decreasing order by relevance, the greedy approach gets close to finding optimal rankings without much exploration.

Compared to the greedy approach, the beam search performs closer to the \textit{ex-post} trade-off values. 
Especially at lower $\alpha$ values, it can achieve lower unfairness while having a comparable utility to the greedy approach. Since the beam search evaluates trade-off values for more documents for ranking, it ``explores'' more permutations which can avoid the myopic choices of only placing the highest relevant document at the top position. 
 Naturally, this exploration depends on the beam width $B$, and in our evaluation, we see the best results using $B=4$.
Beam search does not completely avoid choosing suboptimal documents at intermediate steps when unfairness dominates utility (lower $\alpha$ values), and finding the perfect fairness for ranking remains a combinatorial problem~\cite{bektasUsingLpnormsFairness2020}. \looseness=-1

\textbf{COMPAS Dataset.}
Unlike the synthetic dataset, where we know the true relevance of each document, in real-world scenarios, we are not given access to the true relevances at test time. Therefore, we need to learn the relevances from the dataset and use them for the beam search and greedy approaches. 
To do so, we train a Gradient Boosting Regressor to predict the \texttt{violent\_decile\_score} from the document features as $\rel$. The regression model achieves an $R^{2}=0.79$ and we use these scores for beam search and greedy approaches.
For the PL-RELAX approach, we learn a linear model to map features to the Plackett-Luce coefficients. This gives PL the same representation power as the regression model used to learn relevances for greedy and beam search.

Fig.~\ref{fig:alt-scoring-compas} shows the (test) trade-off curves traced by the beam search, greedy, and PL-RELAX. 
Due to the larger size of the query in this case ($m$=10) compared to the synthetic dataset ($m$=8), we do not calculate the \textit{ex-post} trade-offs.
The beam search manages to achieve high utility and low unfairness at $\alpha < 0.2$ and remains weakly-dominant at $\alpha \geq 0.2$ values.
The greedy approach gets closer to the beam search trade-off values as $\alpha$ gets larger, but it maintains a gap with beam search until the highest $\alpha$ values.
At $\alpha > 0.5$ values, the trade-off values for the beam search and greedy approaches become similar. The PL model, on the other hand, does not achieve as significant of a trade-off as the other two approaches for all $\alpha$ values. It sacrifices utility for lower unfairness at $\alpha < 0.2$, but is weakly-dominated by the greedy and beam search approaches at $\alpha \geq 0.2$ and approaches the beam search and greedy trade-off values at $\alpha = 1$. \looseness=-1




We emphasize that these experiments use the greedy and beam approaches as representatives of ex-post techniques from the literature and the Plackett-Luce model as a representative of scoring-based techniques, rather than competitors of these techniques. The key takeaway is that, even in practice, scoring-based ranking struggles to achieve the same quality of trade-offs that ex-post approaches can, giving empirical evidence to the issue highlighted by the theory. 




\vspace{-6pt}
\section{Conclusion and Limitations}
\vspace{-6pt}
In this paper, we looked closely at the problem of learning to rank in the presence of a traditional utility objective as well as a secondary objective, namely fairness. We challenged the common paradigm of \textbf{learning a score} per query-document pair, by giving several counterexamples that show that this \textbf{can fall short of achieving good trade-offs} between the two objectives. The key insight is that, despite being an in-processing technique, scoring is limited by the fact that \textbf{fairness is non-decomposable in contrast to utility}. We showed that our approximate ex-post solutions, in the form of \textbf{greedy and beam-search, are tractable alternatives with much better trade-offs}. They both, however, rely on learning relevances directly. This opens the question of, if other properties are learned during in-processing instead of or in addition to relevances, whether they could promote even better trade-offs.

On the limitations front, it is worth noting that, instead of interaction disparity, there are some fairness notions that are presented as direct properties of the score function, designed to encourage certain ranking behaviors. For example, \citep{zhu2020measuring} makes the score weakly predictive of the group membership, in order to encourage being at chance when guessing the group of documents placed at higher positions via scoring. This does not map to an interaction disparity as covered in this paper, e.g., it does not account for relevances. The conclusions of this paper do not apply to such notions, and it is entirely plausible that when fairness is a property of the score, it is sufficient to perform in-processing to achieve optimal utility-fairness tradeoffs. In terms of other limitations, while the framework we present readily handles the commonly used DCG as utility (the counterexample constructions would also work with nDCG), there are some utilities such as MAP (mean average precision) that do not fit in the framework. That said, since the shortcomings of scoring are due to the fairness function, we expect tradeoffs involving utilities such as MAP to also have optimality gaps under scoring.

\bibliographystyle{unsrtnat}
\bibliography{ref}  

\balance
\newpage
\onecolumn

\title{Scoring is Not Enough: Addressing Gaps in\\Utility-Fairness Trade-offs for Ranking\\(Supplementary Material)}

\maketitle

\appendix

\section{Related Work}
\label{sec:related-work}

\subsection{Scoring Functions}
\label{sec:rw-scoring-funcs}
Scoring is one of the most traditional approaches to ranking documents. It works by assigning a score to individual documents and then using the score to sort them~\cite{agarwalAddressingTrustBias2019}. A scoring function, often learned, produces the score using the document features and a query. Typically, the scores are expected to help rank documents in such a way that the most relevant ones are shown first~\citep{joachimsOptimizingSearchEngines2002,burgesLearningRankNonsmooth2006}. 
The scoring functions are evaluated using a utility objective, which is defined over cumulative relevances of ranked documents for a query, weighted by the placement or position. Prior works have explored learning methods using three different types of utilities: $(i)$ pointwise --- scores are learned by comparing against labeled relevances of individual documents~\citep{geyInferringProbabilityRelevance1994}; $(ii)$ pairwise --- scores are learned based on relative relevances for document pairs~\citep{burgesLearningRankUsing2005,beutelFairnessRecommendationRanking2019}; $(iii)$ listwise --- scores are learned based on the complete list of documents presented~\citep{caoLearningRankPairwise2007,burgesLearningRankNonsmooth2006}. 
Recent work focuses more on the listwise utility functions, since they allow direct optimization over the complete ranked documents and are easily compatible with desirable end-to-end performance measures~\citep{caoLearningRankPairwise2007}. Despite the utility being listwise, the scoring is performed document-by-document. Thus, the listwise utility imparts valuable statistical information about other documents in a typical query to the score --- but no direct information about other documents in a query is available to the scoring function. One of the common listwise techniques to train learning-to-rank (LTR) methods uses the Plackett-Luce (PL) model to output a ranking based on the probability of selecting the next document conditioned on observed documents. The PL model uses probabilistic sampling instead of a non-differentiable sorting function, making the model differentiable and suitable for gradient-based training methods. However, in practice, calculating gradients for the PL model requires sampling all possible permutations given a document set, making the problem computationally intractable~\citep{schulmanGradientEstimationUsing2015}.
To solve this, researchers have proposed several techniques to efficiently estimate unbiased gradients for the PL model~\citep{oosterhuisComputationallyEfficientOptimization2021,maLearningtoRankPartitionedPreference2021,huijbenReviewGumbelmaxTrick2022,koolEstimatingGradientsDiscrete2020}.
In particular, we use the low-variance gradient estimation technique proposed by \citeauthor{gadetskyLowvarianceBlackboxGradient2019} to compute PL model gradients and optimize for utility-fairness trade-off~\citep{gadetskyLowvarianceBlackboxGradient2019}. This is inspired by reinforcement learning and is in line with many recent approaches for learning to score~\citep{singhPolicyLearningFairness,ge2022toward}.

\subsection{Fairness in Ranking}
As ranking systems are increasingly adopted within sociotechnical systems, fairness considerations have become a crucial requirement. Despite scoring functions having originated from utility-based objectives, significant research has been centered around developing fairness-focused scoring functions that are used for ranking~\citep{singhFairnessExposureRankings2018,asudehDesigningFairRanking2019,singhPolicyLearningFairness,memarrastFairnessRobustLearning2021}.
It is worth describing the landscape of fairness, to place the current work more clearly within context. In addition to the subtle variations in the definitions of fairness, there are distinct granularities, targets, and scopes. In terms of granularity, and a distinction that permeates the general algorithmic fairness literature, there is a contrast between individual notions of fairness~\citep{biegaEquityAttentionAmortizing2018,saitoFairRankingFair2022,bowerIndividuallyFairRankings2020} and group notions~\citep{yangMeasuringFairnessRanked2017,singhPolicyLearningFairness,singhFairnessExposureRankings2018}. In the present work, we focus on binary groups. In terms of targets, since there are two sides in the ranking scenario (query/user, document/item), fairness considerations could be toward one, the other, or both. We focus on document-group fairness. Some of the cited work also considers fairness toward queries/users instead and more recent work also attempts to be jointly fair toward both queries/users and documents~\citep{greenwood2024user}. Document (un)fairness is often measured in terms of interaction disparity, and this can be measured for each query~\citep{gorantlaOptimizingGroupFairPlackettLuce2023} or by accumulating interaction across multiple queries~\citep{memarrastFairnessRobustLearning2021,ovaisiFairnessInteractionRanking2024}. We primarily focus on the former, though we also give a brief treatment of the latter. For interested readers, an extensive review by \citeauthor{zehlikeFairnessRankingPart2022} elaborates on the recent advances in fair ranking methods, datasets used and their applications~\citep{zehlikeFairnessRankingPart2022, zehlikeFairnessRankingPart2022a}. 

\subsection{Post-Processing}
Scoring is not the only way that fairness in ranking has been tackled. Though not as widespread, post-processing --- also known as ex-post methods or re-ranking --- attempts to account for fairness at inference-time. Some of the work in this area does not address trade-offs, but rather attempts to achieve a fixed notion of fairness via post-processing~\citep{geyik2019fairness,gorantlaOptimizingGroupFairPlackettLuce2023}, e.g., based on alternations between groups of documents. One important concept from~\citep{geyik2019fairness} is the use of greed to address the intractability of post-processing. A novel concept in~\citep{gorantlaOptimizingGroupFairPlackettLuce2023} is that scores are learned in a way that is aware of post-processing. Two papers that do address utility-fairness trade-off are~\citep{gao2020toward} and~\citep{ovaisiFairnessInteractionRanking2024}. They both also use greed to develop tractable heuristics. The former uses an entropy-based diversity notion of fairness that does not fit well into the present framework of fairness. The latter does fit, but does not bring to light the fact that scoring is otherwise suboptimal. The post-processing approaches of Section \ref{sec:evaluation} are offered as representatives of, rather than competitors to, these methods.

In summary, we adopt an interaction-disparity notion of fairness over document-groups, and study the trade-offs between utility and such fairness, comparing scoring and non-scoring based ranking methods. Our focus is on the scope of individual queries, but we also briefly address the case when the scope is across queries. What has gone under the radar in prior work, and what we highlight at present, is the inadequacy of scoring in terms of achieving the best possible trade-offs. 


\subsection{Popularity of Scoring in Practice}

The popularity of scoring-based methods in research stems directly from their extensive application across various domains, as scoring is an inherent mechanism in many real-world ranking systems. Although our present contribution is theoretical, the widespread use of scoring in practice, including technology, social sciences, and healthcare, underscores its real-world relevance.

\begin{enumerate}
    \item \textbf{Technology.}
    \begin{enumerate}
        \item LinkedIn. The LinkedIn Talent Search platform, used by recruiters to find eligible candidates, relies on score-based ranking to order results~\cite{behzadExternalFairnessEvaluation2025,ha-thucQueryByKeywordQueryByExampleLinkedIn2017}.
        \item Airbnb. This online accommodation marketplace uses list ranking based on scoring methods to personalize property suggestions for users~\cite{grbovicRealtimePersonalizationUsing2018a}.
    \end{enumerate}
    \item \textbf{Housing.}  Scoring is a foundational component for allocating crucial resources, such as housing services for unhoused community members, as demonstrated by the Allegheny County Department of Human Services~\cite{chengAlgorithmAssistedDecisionMaking2024}. Critically, fairness issues within these scoring methods have a heavy bearing on the reliability and funding opportunities for public agencies.
    \item \textbf{Healthcare.} US kidney exchange programs utilize a point-based ranking system to match and prioritize patients eligible for transfers, balancing key criteria such as medical efficiency and fairness~\cite{bertsimasFairnessEfficiencyFlexibility2013}.
\end{enumerate}

These are some of the many examples we list emphasizing the practical use of scoring-based methods for ranking, where fairness is an important consideration. This should demonstrate the importance of our findings and the urgency to discuss this with the wide community of researchers and practitioners.

\section{Proofs} \label{app:proofs}

\subsection{Proof, Corollary, and Example of Theorem \ref{thm:weak_scorable}}

To prove the theorem, we first build a stepping stone by showing an even more stringent notion of scorability is not possible to achieve. In this notion, one expects the score to rank \emph{every} $\query$ the best way possible.

\begin{definition}[Strong Scorability] \label{def:strongly-scorable}
    A pair of utility and unfairness functions $(\util,\unfair)$ is called deterministic (or, respectively, randomized) \emph{strongly $\alpha$-scorable} if $\forall$ problem instances $(\docset, \rel, \group)$, $\exists$ deterministic (or, respectively, randomized) $\score$ such that $\forall \query = \{d_1, d_2, ..., d_m\}$, ranking them by score using $\sigma_{\score} = \mathsf{sort}_\downarrow(\query, \score)$ achieves
    \begin{equation}  \label{eq:strong-optimality}
    \begin{multlined}
        \tradeoff_\alpha\left(\sigma_{\score}, \query\right) =  \max_{\sigma_\query} \alpha \cdot \util(\sigma, \query) - (1-\alpha) \unfair(\sigma, \query),           
    \end{multlined}
    \end{equation}
    where the maximization is over deterministic (or, respectively, randomized) rankings.
\end{definition}

Unlike Eq. \eqref{eq:weak-optimality}, the maximization variable in Eq. \eqref{eq:strong-optimality} does not have to explicitly depend on $\query$, since this is repeated for every $\query$.

Just as in Definition \ref{def:weakly-scorable}, by allowing $\score$ to depend on the problem instance, Definition \ref{def:strongly-scorable} abstracts the notion that the scoring function may be learned given sufficient data from this instance. Again, the order of quantifiers ``$\exists \score$ s.t. $\forall \query$'' is important for this abstraction to be meaningful. (This should not be confused with the reversed order ``$\forall \query$ $\exists \score$'', which can be trivially achieved by setting the scores to reflect the post-hoc ranking in Eq. \eqref{eq:strong-optimality}.)

What makes the requirement of Definition \ref{def:strongly-scorable} ``strong'' is that we ask for optimality of the trade-off for every query $\query$. In other words, this means that one can think of nature as an adversary who can, within the problem instance, choose the specific documents to present to the scorer. 

This stronger requirement gives us the first evidence that we should not take scoring for granted, through the following suboptimality of deterministic scoring in terms of strong scorability.

\begin{lemma}[Counterexample to Deterministic Strong Scorability] \label{thm:strong_scorable}
    Let $\mathcal U$ and $\mathcal V$ be given by Eqs. \eqref{eq:utility} and \eqref{eq:unfairness} with some monotonically non-increasing $w$ such that $w_i+w_{m-i+1}$ is not constant for $i=1,\cdots,m-1$. Then, for any $\alpha<1$, there exists a problem instance $(\docset, \rel, \group)$ such that $\forall$ deterministic scoring functions $\score$ that can depend on the instance, there exists $\query$ such that Eq. \eqref{eq:strong-optimality} fails. Therefore, this choice of $(\mathcal U, \mathcal V)$ is not strongly scorable.
\end{lemma}
\begin{proof}

    Consider the function \[\phi(j) = \frac{1}{j}\sum_{i=1}^j w_i - \frac{1}{m-j}\sum_{i=j+1}^m w_i,\] with $j=1,\cdots,m-1$. Observe that, thanks to the non-increasing nature of $w$, $\phi$ is non-negative. We claim that $\phi$ cannot be symmetrical, i.e., $\phi(j)=\phi(m-j)$ $\forall j$, under the given condition for $w$. To see this, note that:
    \[
        \phi(j) - \phi(m-j) = \frac{1}{j} \sum_{i=1}^j c_i -\frac{1}{m-j} \sum_{i=j+1}^m c_i,
    \]
    where $c_i=w_i+w_{m-i+1}$. If we enumerate these equations over all $j=1,\cdots,m-1$ and equate them to $0$, we get $m-1$ linearly independent equations, if we add to these $c_1=c$, we have $m$ equations that uniquely determine the values of $c_i$. Since $c_i=c$ $\forall i$ is a valid solution, it is the only one. (We can also solve this constructively, by induction.) Thus $\phi(j)=\phi(m-j)$ $\forall j$ implies that $w_i+w_{m-i+1}$ is constant. (For completion, note that the converse can also be verified to be true, thus these two conditions are equivalent.)

    Since $w_i+w_{m-i+1}$ is \emph{not} constant, it  means that there exists some $p \in \{1,\cdots,m-1\}$ such that $\phi(p) > \phi(m-p)$ (strictly). Using this, construct the problem instance as follows. Let $\docset$ consist of two types of documents, with $(\rel(d),\group(d))=$ either $(1,a)$ or $(1,b)$. Note that, by choosing both types to have the same relevance, the utility of all rankings will be the same, and the trade-off will only depend on the unfairness term.
    
    As previously discussed, we assume that $\score(d)$ is a function of $(\rel(d),\group(d))$ and that it cannot assign the same value to the two distinct types of documents. W.l.o.g., assume $\score(1,a) >\score(1,b)$. Construct $Q$ to have $p$ documents of type $(1,a)$ and $m-p$ of type $(1,b)$. This means that the score-induced ranking assigns the top $p$ spots to documents from $G_a$ and the remaining $m-p$ to $G_b$. Note that the unfairness of this ranking is equal to $\phi(p)$. This, by construction, is larger than $\phi(m-p)$, which is the unfairness we would have obtained had we placed the $G_b$ documents first and $G_a$ documents last. Therefore, for this choice of $Q$, Eq. \eqref{eq:strong-optimality} fails and it follows that strong scorability does not hold.
\end{proof}

The counterexample in the proof of this theorem explicitly constructs a hard query $\query$ to counteract each choice of $\score$. This is possible since, in strong scorability, the adversary of the scorer knows exactly which documents are being scored. We are now ready to prove Theorem \ref{thm:weak_scorable}. Now, picking individual $\query$'s to counteract each choice of $\score$ is no longer option, since $\score$ is allowed to depend on how queries are chosen. However, we show that any choice of $\score$ will be suboptimal on \emph{some} queries.

\begin{proof}[Proof of Theorem \ref{thm:weak_scorable}]
    Construct the problem instance similarly to the proof of Lemma \ref{thm:strong_scorable}. Choose $\docset$ to have two types of documents $(\rel(d),\group(d))=(1,a)$ and $(\rel(d),\group(d))=(1,b)$. Let $\D$ sample from each group equally. Let $p$ be as before, and assume, w.l.o.g., that $\score(1,a) >\score(1,b)$. Now let's look at the difference between the right and left hand sides of Eq. \eqref{eq:weak-optimality}:
    \begin{eqnarray*}
         \max_{\sigma_\query} \Ex_\query\left[\tradeoff_\alpha\left(\sigma_Q, \query\right)\right] - \Ex_\query\left[\tradeoff_\alpha\left(\sigma_{\score}, \query\right) \right]
         &=& \max_{\sigma_Q} - (1-\alpha) \Ex_Q\left[ \unfair(\sigma_Q, Q)\right] + (1-\alpha) \Ex_Q\left[ \unfair(\sigma_{\score}, Q)\right] \\
        &\geq& - (1-\alpha) \Ex_Q\left[ \unfair(\tilde \sigma_Q, Q)\right] + (1-\alpha) \Ex_Q\left[ \unfair(\sigma_{\score}, Q)\right] \\
        &=&  (1-\alpha) \Ex_Q\left[ \unfair(\sigma_{\score}, Q) - \unfair(\tilde \sigma_Q, Q)\right],
    \end{eqnarray*}
    where the utility terms cancel out and we have bounded the optimality gap by choosing a specific ranking $\tilde \sigma_Q$ rather than the optimal.
   


    Let $A$ be the set of all queries $Q$ which contain $p_Q$ documents of $G_a$ and $m-p_Q$ documents of $G_b$, such that $\phi(p_Q)>\phi(m-p_Q)$. We already showed that this set is non-empty, and thus there is a positive probability that $Q$ will be sampled from it. Since we have full control of $\tilde \sigma_Q$ for each $Q$, let $\tilde \sigma_Q$ reverse the order of $\sigma_\score$ or all $Q\in A$ and let it rank similarly to $\sigma_\score$ otherwise (thus the unfairness difference is $0$ in that case).

    It then follows that:
    \begin{eqnarray}
         \max_{\sigma_\query} \Ex_\query\left[\tradeoff_\alpha\left(\sigma_Q, \query\right)\right] - \Ex_\query\left[\tradeoff_\alpha\left(\sigma_{\score}, \query\right) \right]
         &\geq &  (1-\alpha) \Ex_Q\left[ \mathbf{1}\{Q\in A\} \left(\phi(p_Q)-\phi(m-p_Q) \right) \right] \label{eq:general_gap}\\
         &>& 0, \nonumber
    \end{eqnarray}
    where we have used the fact that for $Q\in A$, using the proof of Lemma \ref{thm:strong_scorable}, $\unfair(\sigma, Q)$ is not minimal with the scoring-induced ranking. Therefore, we have a non-zero optimality gap. In other words, Eq. \eqref{eq:weak-optimality} fails, and it follows that scorability does not hold.
   
\end{proof}

Note that if there is a slight discrepancy in relevance between the document types, the gap would still exist as long as utility is no so dominant that fairness is outright ignored. We can also get a quantifiable larger gap, under a convexity assumption on the weights.

\begin{corollary} \label{cor:gaps} Consider the setting of Theorem \ref{thm:weak_scorable}. Assume, additionally, that $w_i$ are convex, i.e.. $(w_{i+1}-w_{i}) - (w_{i}-w_{i-1}) > 0$, and let $m$ be odd. Then, if we uniformly sample $m$ documents from two types of documents $(1,a)$ and $(1,b)$, the trade-off gap of scoring can be bounded by
\begin{eqnarray*}
    \max_{\sigma_\query} \Ex_\query\left[\tradeoff_\alpha\left(\sigma_Q, \query\right)\right] - \Ex_\query\left[\tradeoff_\alpha\left(\sigma_{\score}, \query\right) \right]
    &\geq& (1-\alpha) \sum_{0<p<\frac{m}{2}} \P(\mathsf{Binomial}(m,\tfrac{1}{2})=p) \left[\phi(p)-\phi(m-p)\right] \\
    &\geq& (1-\alpha)\left(\frac{1}{2}-\frac{1}{2^m}\right) \underbrace{\left[\phi\left(\frac{m-1}{2}\right)-\phi\left(\frac{m+1}{2}\right)\right]}_{>0}
\end{eqnarray*}
\end{corollary}

\begin{proof}
    What simplifies things in this corollary is that the convexity of $w$ implies that $\phi$ is non-increasing (further properties about $\phi$ can also be shown, such as convexity, but that's not needed.) To see this, let $d_i=w_i-w_{i+1}$ be the differences of weights, as (non-negative) magnitude of consecutive drops in the sequence. The convexity of $w$ implies that $d_i$ is non-increasing, i.e., $d_1\geq d_2 \geq \cdots$. We can express $w_i$ around $w_{j+1}$ by summing these differences as:
    \[
        w_i = \left\{ \begin{array}{lcl}
        w_{j+1} + \sum_{k=i}^j d_k &;& i\leq j\\
        w_{j+1} - \sum_{k=j+1}^{i-1} d_k&;& i > j+1
        \end{array}
        \right.
    \]
        
    By changing orders of summation, we have:
    \[
        \sum_{i=1}^j w_i = \sum_{i=1}^j \left(w_{j+1} + \sum_{k=i}^j d_k \right) = j~ w_{j+1} + \sum_{k=1}^j k~d_k  
    \]
    and
    \[
        \sum_{i=j+1}^m w_i = w_{j+1} + \sum_{i=j+2}^m \left(w_{j+1} - \sum_{k=j+1}^{i-1} d_k \right) = (m-j)~w_{j+1} - \sum_{k=j+1}^{m-1} (m-k)~d_k. 
    \]

    By combining these, we get
    \[
        \phi(j) = \frac{1}{j} \sum_{i=1}^j w_i - \frac{1}{m-j}\sum_{i=j+1}^m w_i = \frac{1}{j} \sum_{k=1}^j k~d_k  + \frac{1}{m-j} \sum_{k=j+1}^{m-1} (m-k)~d_k.
    \]

    We now can deduce that $\phi$ is non-increasing, because
    \begin{eqnarray*}
        \phi(j)-\phi(j+1) &=& \left(\frac{1}{j}-\frac{1}{j+1} \right)\sum_{k=1}^j k~d_k -\frac{j+1}{j+1}~d_{j+1} + \frac{m-j-1}{m-j} d_{j+1} \\
        &&  + \left(\frac{1}{m-j}-\frac{1}{m-j-1} \right)\sum_{k=j+2}^{m-1} (m-k)~d_k \\
        &=&\frac{1}{j(j+1)} \sum_{k=1}^j k~d_k - \frac{1}{(m-j)(m-j-1)} \sum_{k=j+1}^{m-1} (m-k)~d_k \ \geq\  0,
    \end{eqnarray*}
    where for the non-negativity we used the fact that the $d_k$'s are non-increasing and the fact that we are evaluating the difference (up to a factor of $2$) of the weighted averages of the first $j$ $d_k$'s and the remaining (no greater) $d_k$'s.

    Because $\phi$ is non-increasing, it implies that the symmetric differences $\phi(p)-\phi(m-p)>0$ for all $p<\frac{m}{2}$. What we are thus able to achieve is to exactly specify the whole set $A$ in the proof of Theorem \ref{thm:weak_scorable}. Using the same problem instance there, instead of simply claiming that Eq. \eqref{eq:general_gap} is non-negative we can evaluate it to get the first bound in the claim of this corollary.

    To get the second bound, we note that the fact $\phi$ is non-increasing also implies that the symmetric differences $\phi(p)-\phi(m-p)$ are non-increasing too. Therefore, we can replace them by the smallest value they take $\phi\left(\frac{m-1}{2}\right)-\phi\left(\frac{m+1}{2}\right)$, at $p=\frac{m-1}{2}$. The proof is then complete by pulling that out of the sum, which gives us $\frac{1}{2}$ (the median of a symmetric binomial distribution) minus the chance of $p=0$, when there cannot be any unfairness, by convention.

    To put this in context, note that without the ability to characterize $A$, the only bound we can offer for the gap is exponentially smaller, obtained by multiplying with the probability of a sequence with exactly $p$ documents of the needed type, giving us:
    \[
        \frac{1}{2^m} \binom{m}{p} \left(\phi(p)-\phi(m-p)\right). \qedhere
    \]
\end{proof}

\paragraph{Example.} As an example to illustrate this corollary, consider the weights $w_i=\frac{1}{\log_2(i+1)}$ for position bias. These satisfy the requisite assumptions, including convexity. If we use queries of $m=9$ documents and use the simplified (second) bound we can easily find that $\phi(4)-\phi(5)\approx 0.04$ and thus we deduce that all scoring functions would have an optimality gap of at least $0.02$ (times $1-\alpha$). We can refine this further by using the full (first) bound, and find out that the gap is in fact at least $0.05$ (times $1-\alpha$).

\begin{figure}[h]
    \centering
    \includegraphics[width=0.8\linewidth]{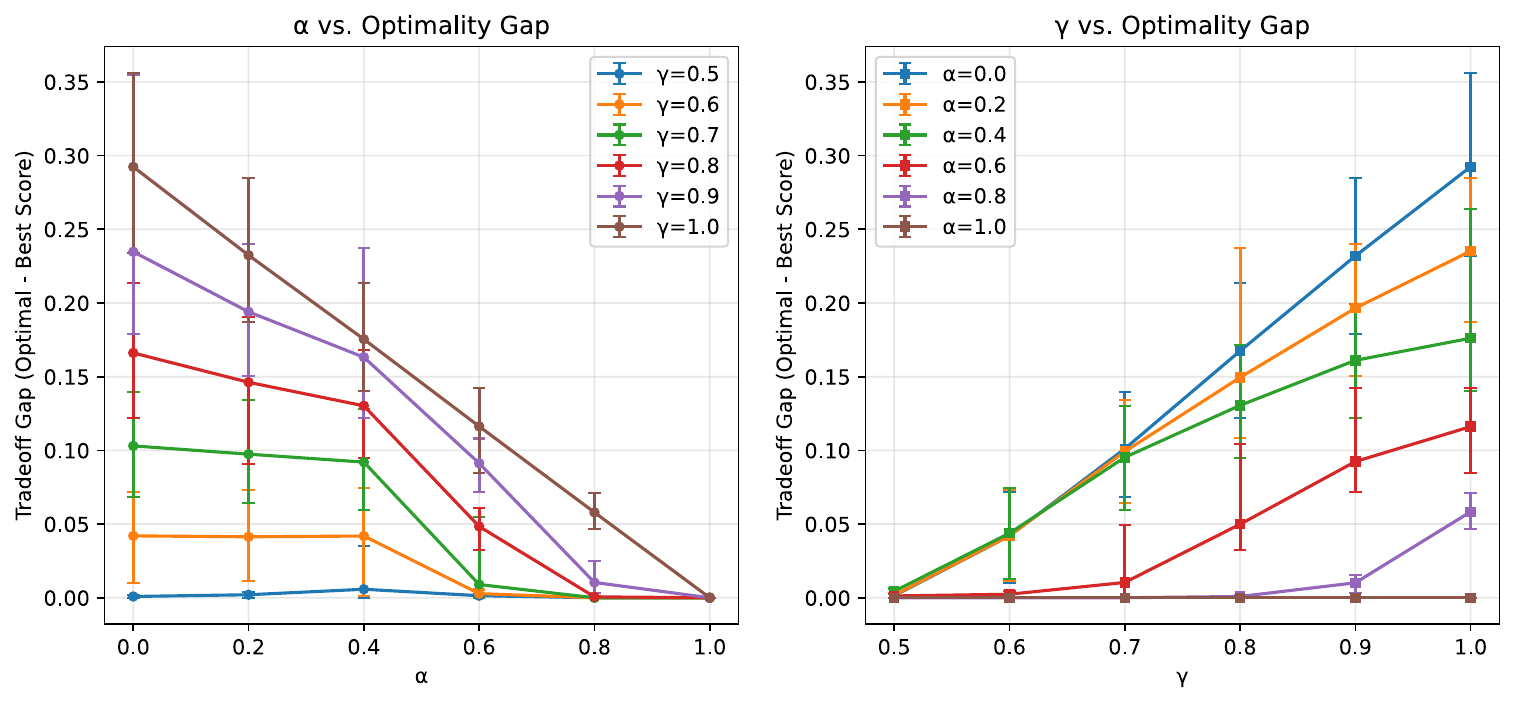}
    \caption{Tradeoff gap between the optimal ranker and the best scorer in the example of Corollary \ref{cor:gaps}, as a function of the desired tradeoff parameter $\alpha$ and asymmetry of relevances $\gamma$ between the two types of documents.} \vspace{10pt}
    \label{fig:corollary-gaps}
\end{figure}

Of course, in practice the gap is even larger than what this theoretical analysis predicts, because the optimal ranking may interleave the documents but a deterministic scorer cannot. (The theory is only comparing against an alternative that also does not interleave.) We can simulate this scenario straightforwardly, and in doing so, we also add an asymmetry $\gamma$ in the relevances between the two groups (one group with relevance $1$, the other with relevance $\gamma$), to assess the impact of utility. In Figure \ref{fig:corollary-gaps}, we illustrate how these two parameters affect the gap. The closest to the theory is the left figure, with $\gamma=1$ and the gap plotted vs. $\alpha$. We can see the expected linear $1-\alpha$ behavior, though the gap is about $6\times$ larger than the theory's lower bound of $0.05$ (it's $\approx 0.30$ at $\alpha=1$.)

As $\gamma$ increases, the gap diminishes because the utility plays a larger role. With $\gamma$ very asymmetrical, utility dictates the tradeoff by requiring to rank all the relevance-$1$ documents ahead, and the gap vanishes. As expected, the diminishing with $\gamma$ is more marked for larger values of $\alpha$, where utility dominates and the overall tradeoff becomes closer to being separable. Indeed, even for $\gamma$ near $1$, the gap can nearly flatten near $\alpha=1$. However, it remains non-negligible whenever $\alpha$ is small enough. It is remarkable that the behavior of this very simple experiment is in complete agreement with our larger set of experiments, including those with real data.

(The experiment still uses $m=9$. Expected utility and unfairness were calculated using 500 random queries. The bars represent 10$^\textrm{th}$ and 90$^\textrm{th}$ percentiles of the gap across these queries.)

\subsection{Proof of Theorem \ref{thm:rand_scorable}}

\begin{proof}[Proof of Theorem \ref{thm:rand_scorable}]

First, consider strong scorability and use the same problem instance $(\docset, \rel, \group)$ as in the proof of Lemma \ref{thm:strong_scorable}. Namely, let $\docset$ consist of two types of documents, with $(\rel(d),\group(d))=$ either $(1,a)$ or $(1,b)$. We will proceed by competing with the score-induced randomized ranking using an alternative randomized ranking. Given any randomized scoring function $\score$, assume for the sake of contradiction that it achieves Eq. \ref{eq:strong-optimality} for all $Q$. Let $X=\score(1,a)$ and $Y=\score(1,b)$ be independent random variables representing a sample of scoring from each document type. Assume, without loss of generality, that $P(X>Y)>0$. Construct $Q$ to have $p$ documents of type $(1,a)$ and $m-p$ of type $(1,b)$. Let $E$ be the event that all $(1,a)$ documents will have greater score than all $(1,b)$ documents. Since scores are sampled independently per document, and based on how $X$ and $Y$ relate, we know that $\P(E)>0$ (possibly small but non-zero).

Consider an alternative randomized ranker that mimics $\sigma_{Q,\score}$ on $E^\mathsf{c}$, but on $E$ it flips the order of all $(1,a)$ and $(1,b)$ documents. This will change the unfairness, given $E$, from $|\phi(p)|$ to $|\phi(m-p)|$, thus reducing it (see the proof of Lemma \ref{thm:strong_scorable}). By the law of total expectation, the expected unfairness of the alternative randomized ranker is thus strictly smaller than that of $\sigma_{Q,\score}$, leading to an improved trade-off and thus a contradiction. Thus $\score$ must fail Eq. \ref{eq:strong-optimality} for this $Q$. Note that the proof is valid for \emph{any} score-based randomized ranker, i.e., none of them can be trade-off optimal for all $Q$. Thus, the trade-off optimal randomized ranker is provably non-score-based.

The proof of scorability follows the same steps as the proof of Theorem \ref{thm:weak_scorable}, but adapted to the above alternative randomized ranker construction. In particular, with positive probability, $Q$ will have $p$ documents of $G_a$ and $m-p$ documents of $G_b$. Conditionally on this $Q$, the event $E$ will again have a positive probability. The alternative randomized ranker will mimic $\sigma_{Q,\score}$ in all cases except under this $Q$ and event $E$, and thus will have an improved trade-off. It follows that scorability once again does not hold.


\end{proof}

\section{Across-Query Fairness} \label{app:across-query}

The notion of unfairness given by Eq. \eqref{eq:unfairness} is per-query. Namely, we measure the discrepancy between groups at the level of individual queries. While this is reasonable in many applications, e.g., suggesting male and female applicants in response to a single job posting, there are other instances where fairness is better measured across queries, e.g., whether certain movie categories are reaching a wide enough audience. We here demonstrate that even in that setting, scoring can be suboptimal.

We assume, as in the case of scorability, that queries are formed based on a distribution $\query \sim \D$. We focus on deterministic scoring. The notion of utility is mostly unchanged, except that we now average over all queries:
\begin{equation} \label{eq:across-utility}
    \overline{\util}(\sigma, \D) = \Ex_{\query \sim \D}\Barg*{\util(\sigma, \query)}.
\end{equation}

The notion of unfairness, however, is more radically affected. Instead of aggregating interaction per query, we aggregate it across queries:

\begin{equation} \label{eq:across-queries-unfairness}
    \begin{multlined} 
    \overline{\unfair}(\sigma, \query) = \left| \Ex_{\query \sim \D} \left[ \frac{1}{m \P (G_a)} \sum_{i:\group(\sigma_\query(i)) = a} \!\!\!w_i \rel(\sigma_\query(i)) \right.\right. \\ \left. \left. - \frac{1}{m \P (G_b)} \sum_{i:\group(\sigma_\query(i))=b} \!\!\!w_i \rel(\sigma_\query(i)) \right] \right|  
    \end{multlined}
\end{equation}
where
\[
    \P (G) = \frac{1}{m} \Ex_{\query \sim \D} \left[ \sum_{i} \mathds{1}\{\sigma_\query (i) \in G\} \right].
\]
Contrast the inside-absolute-value expectation of \eqref{eq:across-queries-unfairness} with the outside-absolute-value expectations of \eqref{eq:expected-unfairness} and \eqref{eq:weak-optimality}.
\begin{definition}[Across-Query Scorability] \label{def:across-query} A pair of across-query utility and unfairness functions is called scorable if
    $\forall \alpha$, $\forall \D$, $\exists \score_{\alpha,\D}$ such that if $\docset = \{d_1, d_2, ..., d_m\} \sim$ i.i.d. from $\D$ are sorted using $\sigma_\score = \mathsf{sort}(\docset, \score_{\alpha,\D})$, we achieve
    \begin{align*} \label{eq:across-query-scorable}
        \max_{\sigma} \quad \alpha \cdot \overline{\util}(\sigma, \D) - (1-\alpha) \cdot \overline{\unfair}(\sigma, \D).
    \end{align*}
\end{definition}

\textbf{Counterexample to Across-Query Scorability.} To demonstrate that even in the across-query setting scorability is not to be taken granted, we construct a specific example where we can tractably solve for the optimal ranking.

In particular, we sample a query $\query=\{\!\{d_1, \ldots, d_m\}\!\}$ from document distribution $\D$ i.i.d. (with replacement), where $\D$ is a categorical distribution of $k$ documents. We assume $m \geq k$. The number of all possible multisets is given $\multiset{k}{m}$. Within each multiset, the number of possible permutations is given by $\binom{m}{m_1,\ldots,m_k}$. Recall that the maximization in Definition \ref{def:across-query} determines, for every multiset $\query$, the particular permutation representing the way they would be ranked. Unlike the per-query instance where the optimization could be solved for each query separately, the across-query optimization needs to take into account all queries simultaneously. 


To solve this, we propose the following linear relaxation. A multiset is equivalently identified via the corresponding histogram $H=\textsf{Histogram}(\query)$. Let $J(H)$ be the set of all permutations within the histogram $H$, and let $j\in J(H)$ represent the indices of individual permutations. In other words, choices of $j$ in each $H$ collectively define $\sigma$. Each $j$ results in a contribution to utility, which we denote by $\mathcal U_{H,j}$, as well as a contribution to the unfairness term within the absolute values, which we denote by $\mathcal V_{H,j}$. These contributions can be computed offline. If we relax the selection of the permutation to a convex combination over permutations given by coefficient $x_{j|H}$, akin to randomized sampling, the maximization of the trade-off can be written as a linear program as follows:
\begin{align*}
    \text{maximize} \quad & \alpha \sum_{H} p(H) \sum_{\substack{j \in J(H)}} x_{j|H} \, \util_{H, j} - Z \\
    \text{subject to} \quad & (1 - \alpha) \sum_{H} p(H) \sum_{\substack{j \in J(H)}} x_{j|H} \, \unfair_{H, j} \leq Z \\
    & -(1 - \alpha) \sum_{H} p(H) \sum_{\substack{j \in J(H)}} x_{j|H} \, \unfair_{H, j} \leq Z \\
    & \sum_{\substack{j \in J(H)}} x_{j|H} = 1, \quad \forall H \\
    & x_{j|H} \geq 0, \quad \forall H, j \\
    & Z \geq 0 \,.
\end{align*}
We choose a uniform distribution for $\D$, which then results each sequence to be equally likely. The probability $p(H)$ for a histogram $H$ is then computed using the multinomial coefficient, representing the number of ways to arrange elements according to the counts in $H$:
\begin{align*}
    p(H) = \frac{m!}{h_1! \, h_2! \, \dots \, h_k!} \left( \frac{1}{k} \right)^m.
\end{align*}

As for the PL model, we choose $k = 5$, $m = 8$, with the following composition of documents $\left(\rel(d),\group(d)\right) = (1,a), (2,b), (3,a),$ $(4,b), (5,a)$. Weights are chosen to be the standard position bias $w_i = 1 /(\log_2(i)+1)$. This is the same as the synthetic data set described in Section \ref{sec:evaluation}. To compare with scoring, we optimize over all possible deterministic scorings by considering all $k!$ different orderings that they induce on the documents. The results are in Table \ref{tab:scoring-vs-lp}, highlighting a clear gap between scoring and the optimal trade-offs. It is worth noting that despite the linear program being a relaxation, it predominantly produces pure strategies, i.e., a single ranking choice per multiset.

\begin{table}[]
    \centering
    \begin{tabular}{r|rrr|rrr}
    \toprule
        & \multicolumn{3}{l}{\textbf{Optimal Scoring}}                & \multicolumn{3}{l}{\textbf{Linear Program}}                 \\
    $\boldsymbol{\alpha}$ & $\boldsymbol{\util}$ & $\boldsymbol{\unfair}$ & $\boldsymbol{\tradeoff}$ & $\boldsymbol{\util}$ & $\boldsymbol{\unfair}$ & $\boldsymbol{\tradeoff}$ \\ \midrule \midrule
    0   & 11.75 & 0    & 0     & 11.16 & 0    & 0     \\
    0.1 & 11.75 & 0    & 1.18  & 13.21 & 0    & 1.32  \\
    0.2 & 13.27 & 0.33 & 2.39  & 13.21 & 0    & 2.64  \\
    0.3 & 13.27 & 0.33 & 3.75  & 13.21 & 0    & 3.96  \\
    0.4 & 13.27 & 0.33 & 5.11  & 13.21 & 0    & 5.28  \\
    0.5 & 13.27 & 0.33 & 6.47  & 13.21 & 0    & 6.61  \\
    0.6 & 13.36 & 0.45 & 7.84  & 13.21 & 0    & 7.93  \\
    0.7 & 13.36 & 0.45 & 9.22  & 13.21 & 0    & 9.25  \\
    0.8 & 13.36 & 0.45 & 10.6  & 13.21 & 0    & 10.57 \\
    0.9 & 13.36 & 0.45 & 11.98 & 13.21 & 0    & 11.89 \\
    1   & 13.36 & 0.45 & 13.36 & 13.36 & 0.45 & 13.36 \\ 
    \bottomrule
    \end{tabular}
    \caption{Across-query utility, unfairness, and trade-off: suboptimality of scoring vs. direct optimization.}
    \label{tab:scoring-vs-lp}
\end{table}

\section{Beam Trade-off Search Algorithm Details} \label{app:beam-details}

\subsection{Mathematical Description}

Beam Trade-off Search (Algorithm \ref{alg:beam-search}) can be mathematically described as follows. We initialize $B$ lists $(\sigma_{ D}^{1}, \dots, \sigma_{ D}^{B})$, each with one of the most relevant documents from each group. Each of these lists, $\mathfrak{b} = 1, \cdots, B$ maintains its own copy of the pair of group-wise sorted documents, $(D^\mathfrak{b}_a,D^\mathfrak{b}_b)$, with the added document $\sigma^\mathfrak{b}_D(0)$ removed. Let $d^{\mathfrak{b}}_\group \in D^\mathfrak{b}_\group$ be sorted in decreasing order of relevance with increasing $\ell=1,\cdots,|D^\mathfrak{b}_\group|$. For subsequent positions $i > 1$, we collect the most highly relevant document from each group, like in greedy, append them to the documents in each list to create temporary ranking orders, and calculate their trade-offs, as follows. The documents to be evaluated can be written as $\mathfrak{D}^\mathfrak{b} = \{d^{\mathfrak{b}}_\group~:~\group \in \{a,b\}\}$. Thus, there are $2B$ partial rankings considered at each step
\[
    \mathfrak{S} = \left\{ \sigma_{ D}^{\mathfrak b'} (\Barg{0, \dots, i-1}) \cup \{d\} : \mathfrak{b'} = 1, \cdots, B, d\in \mathfrak{D}^\mathfrak{b'}, \right\}.
\]
We sort them in decreasing order of their trade-off values and only keep the top-$B$ lists, to maintain our beam width. We can write this as, for $\mathfrak{b} = 1, \cdots, B$:
\begin{align}
    \sigma_{ D}^{\mathfrak b} &= \underset{\sigma \in \mathfrak{S}}{\textrm{top-}\mathfrak{b} \arg\max} 
        ~\tradeoff\left(\sigma  \right)
\end{align}
We then maintain only the surviving lists' sorted groups, and remove from them each list's newly appended document. We repeat this process until all the remaining documents in the set have been ranked. In the end, the ranking with the best trade-off among the last $B$ rankings is returned.

\begin{figure}[h]
    \centering
    \includegraphics[width=\linewidth]{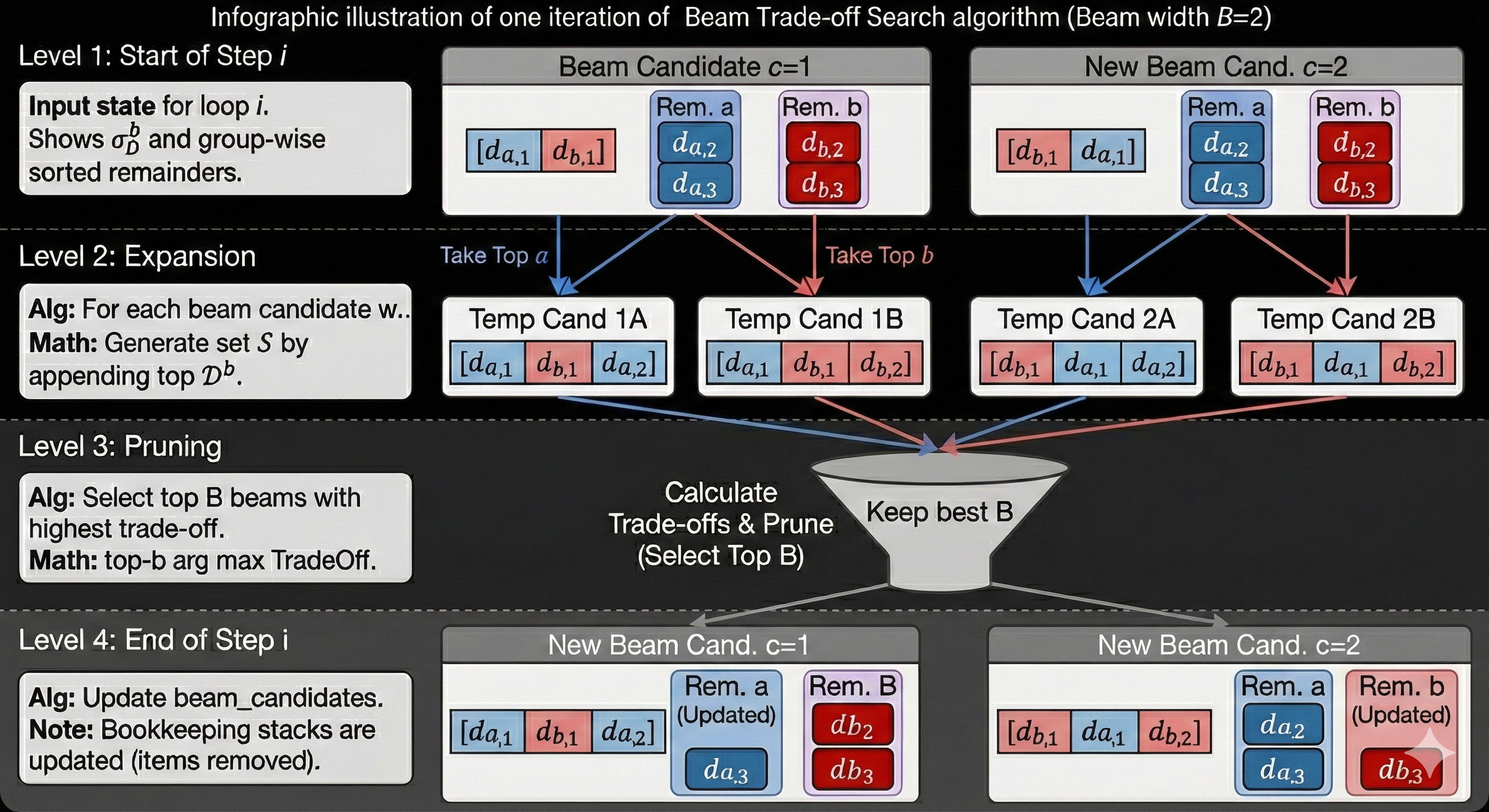}
    \caption{AI-Generated Illustration of Beam Trade-off Search (Algorithm \ref{alg:beam-search})}
    \label{fig:beam-search-ai}
\end{figure}

\subsection{Illustration}

We give an illustration of Beam Trade-off Search. Figure \ref{fig:beam-search-ai} is AI-generated (using Gemini with Nano Banana Pro), based on the \LaTeX pseudocode of Algorithm \ref{alg:beam-search} and the above mathematical description.

\vspace{24pt}

\section{Additional Experiments}
\label{sec:other-datasets}
We now give two additional real data experiments illustrating the presence of optimality gaps for scoring. 

\subsection{German Credit Dataset}
\label{sec:german-credit}
\textbf{Dataset Preparation.}
We follow a similar process to create queries for ranking from the the German Credit dataset as we did from the COMPAS dataset (see~\ref{sec:eval_data}).
The dataset contains 1,000 rows of 20 attributes that comprise an individual's creditworthiness and true relevances. Relevances are binary in this case, i.e., creditworthy ($\rel$=1) or not ($\rel$=0). We maintain the query size, ratio of positive to negative class, and the train-test split as in the COMPAS dataset. The only difference being that we use \texttt{sex=\{Male, Female\}} as the protected attribute in this case.

\textbf{Results.}
Since the relevances in this dataset are binary, we train a Logistic Regression model that maps features to relevances, $\rel \in [0, 1]$. The probability of each individual being creditworthy ($\rel = 1$) is treated as the target variable for this prediction. 
However, the regression model suffers from lower prediction accuracy, which could be explained by the binary relevances compared to multi-valued relevances for the COMPAS dataset.
To counter the effect of noisy learned relevance on the greedy and beam search, we modify the algorithms to sample using a softmax with temperature $\beta=10 \sqrt{\frac{\alpha}{(1-\alpha)}}$ instead of merely sorting group-wise documents based on the relevances. 
Lower values of $\alpha$ make $\beta$ smaller,  making the sample more random and steering the greedy and beam search towards lower utility and lower unfairness. As $\alpha$ increases and $\beta$ gets larger, the softmax sample has a higher probability of being ordered by the learned relevances. 
We train the PL-RELAX model as specified earlier. Since the performance of the trained Plackett-Luce model is not great, in order to rule out optimization-related issues, we also adapt the optimizer of the Fair-PGRank method proposed by ~\citep{singhPolicyLearningFairness}, within their publicly available code, as an additional baseline. Essentially, the difference between the two methods is that Fair-PGRank (with our own objective function, rather than that of \citep{singhPolicyLearningFairness}) uses REINFORCE without the variance reduction of PL-RELAX. The fact that both methods perform similarly gives support to the convergence of the PL-RELAX model on this dataset.

We present the results in Fig.~\ref{fig:alt-scoring-gcredit}. 
Both Plackett-Luce models (PL-RELAX and Fair-PGRank) remain dominated by our ex-post ranking methods (greedy and beam search) for the German Credit dataset. The latter shows excellent trade-offs by maintaining high utility while keeping the unfairness low at all $\alpha$ values. Notably, greedy performs better than the beam search at $\alpha \leq 0.2$ as the myopic choice for greedy results in lower unfairness and similar utility. At $\alpha > 0.2$, beam search provides the optimal trade-offs compared to other approaches and the gap is significant with respect to PL-RELAX and Fair-PGRank.
We observe that the trade-off curves do not look convex for the German Credit dataset, as they did for COMPASS and the synthetic dataset. This is likely due to the fact that this dataset has binary relevances ($\rel\in \{0, 1\}$) and the regression and the scoring-based PL models learn noisy estimates of the relevances, affecting the lower range of $\alpha$ values where the tradeoff curve is steep and small changes in trading off the estimated utility and fairness can result in large changes in the (plotted) true utility and fairness.



\begin{figure}[h]
    \centering
    \includegraphics[width=0.5\linewidth]{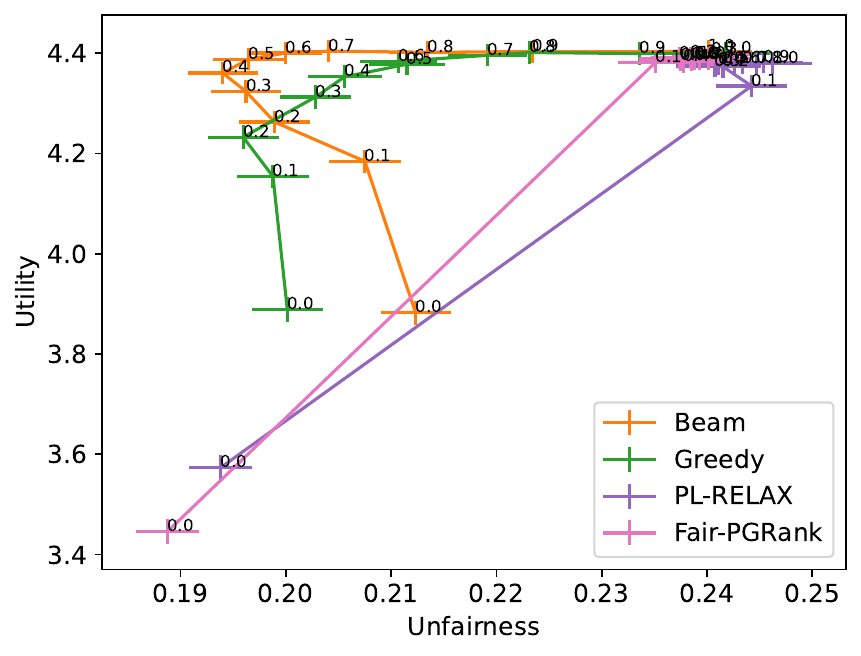}
    \caption{Utility and unfairness values for $\alpha = \{0, \ldots, 1\}$ achieved by the beam search (orange), greedy (green), and the trained PL model (purple) on the German Credit dataset.}
    \label{fig:alt-scoring-gcredit}
\end{figure}

\subsection{MovieLens 100K Dataset}
\label{sec:movielens}
\begin{figure*}[t] 
    \centering
    \begin{subfigure}[t]{0.45\textwidth}
        \includegraphics[width=\linewidth]{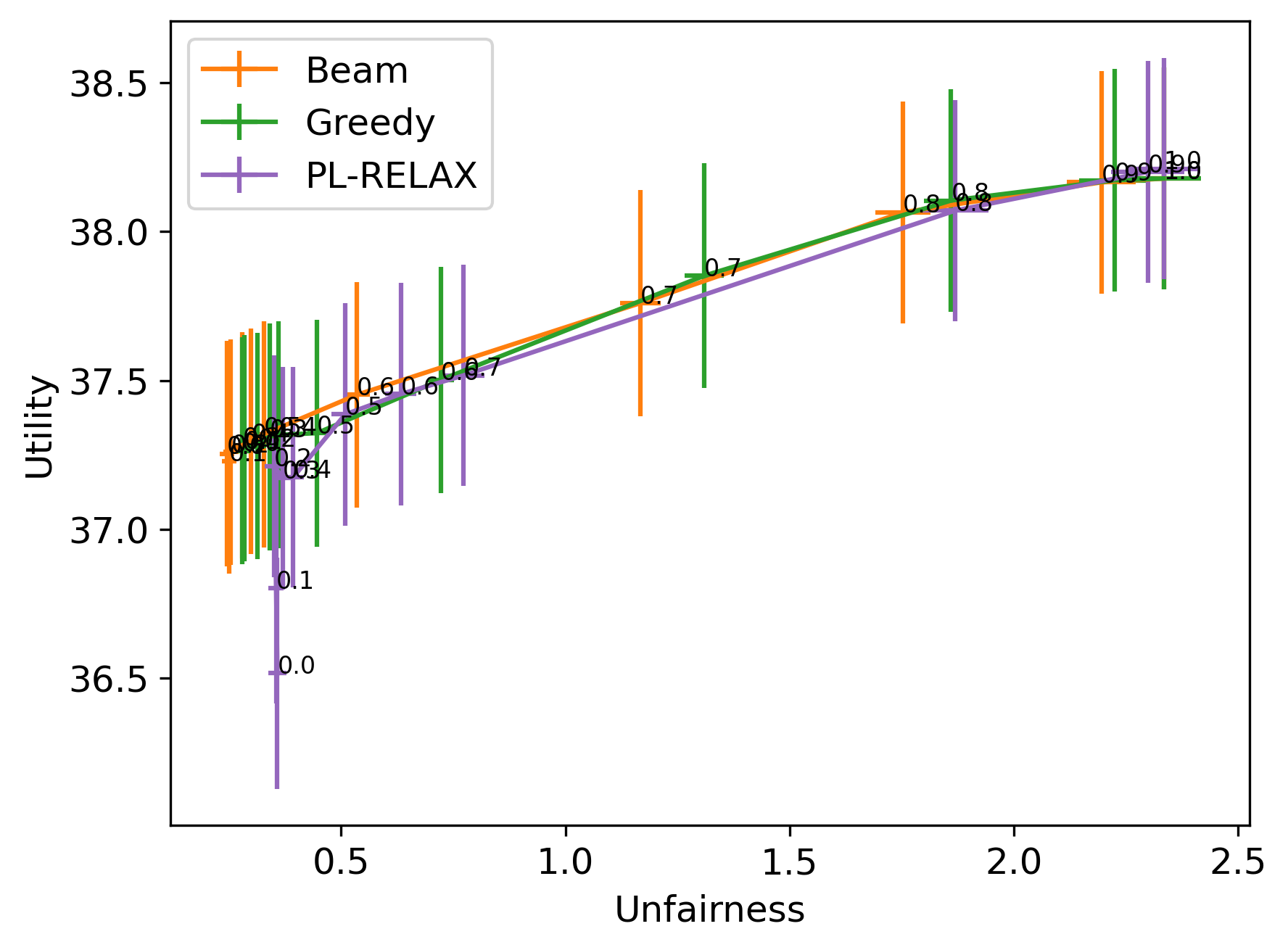}
        \caption{Greedy and Beam Search use the relevances learned from a Linear Regression model.}
        \label{fig:ml100k-linear-model}
    \end{subfigure}
    \hspace{3em} 
    \begin{subfigure}[t]{0.45\textwidth}
        \includegraphics[width=\linewidth]{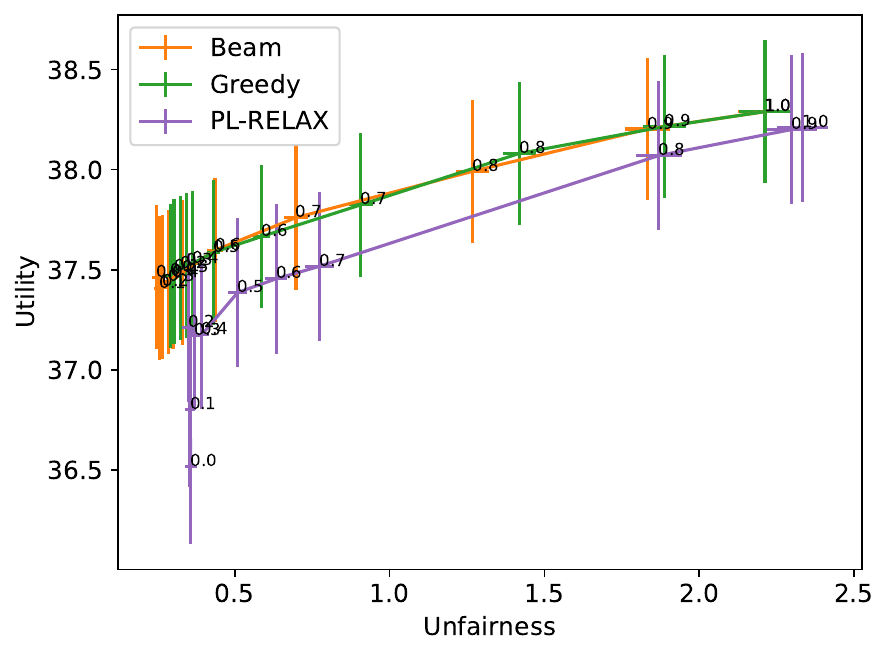}
        \caption{Greedy and Beam Search use the relevances learned from a Gradient Boosting Regression model.}
        \label{fig:ml100k-gb-model}
    \end{subfigure}
    \caption{Utility and unfairness values for $\alpha = \{0, \ldots, 1\}$ achieved by beam search (orange), greedy (green) and the trained PL model (purple) on the Movielens 100K dataset.}
    \label{fig:ml100k-exp}
\end{figure*}

\textbf{Data Preprocessing.} 
MovieLens 100K is a popular dataset used for learning-to-rank tasks~\cite{harperMovieLensDatasetsHistory2016}. As the name suggests, the dataset contains 100,000 movie ratings for 1,682 movies from 943 users. A rating is a number ranging from 1 to 5, where a higher rating is more liked by the user. Following the methodologies of~\cite{naghiaeiCPFairPersonalizedConsumer2022,ovaisiFairnessInteractionRanking2024}, we select the top 20\% of movies with the highest average rating as popular and the rest as unpopular. We use this popularity measure as the protected attribute to calculate the unfairness. For each user, we randomly sample 20 of their rated movies to create a query, such that each query has a fixed size. Based on the 21 features describing a movie, including release year, genre (1-hot encoded), and popularity, we train two regression models to predict point-wise relevances for the movies. Then, we train the PL-RELAX model as described in Section~\ref{sec:experiments}.

\textbf{Results.} 
The dataset is challenging to learn as it contains many sparse features and varied ratings for the same movie among users. Therefore, we explore two regression models to learn ratings from features whose outputs are used by greedy and beam search. First, we train a Linear Regression model with mean squared error loss using an SGD optimizer. The linear model yields $R^{2}=0.09$, suggesting low explanatory power. Second, we use a Gradient Boosting Regression model, an ensemble method that sequentially trains multiple learners to reduce the error of the previous learner. The ensemble model gets $R^{2}=0.11$, showing a minor improvement over the linear model. These are both used as inputs to greedy and beam search approaches. We use the same linear model architecture as for the other datasets to map features to PL coefficients, and we optimize it with the PL model to maximize the trade-off.

The results for both relevance models are shown in Fig.~\ref{fig:ml100k-exp}. Although the linear model demonstrates a weak fit, greedy and beam search performances resemble that of the PL-RELAX model (Fig.~\ref{fig:ml100k-linear-model}). When the slightly more powerful Gradient Boosting model is used for greedy and beam search, they both exhibit better trade-off values at each $\alpha$ level. The challenges of learning from sparse features and varied ratings limit beam search in outperforming greedy search using the same learned relevances. More generally, these results showcase an additional benefit of using decoupled models to learn relevance and then applying our proposed heuristic approaches. 
We can choose better relevance models independently of the ranking objectives, then combine the model outputs with greedy or beam search to achieve a desired utility-fairness trade-off, whereas scoring-based methods like Plackett-Luce must be retrained as the desired trade-off changes.

\end{document}